\documentclass[11pt]{article}
\usepackage{amsmath,amssymb,amsthm}
\usepackage{fullpage}
\usepackage[colorlinks]{hyperref}
\hypersetup{linkcolor=cyan,filecolor=cyan,citecolor=cyan,urlcolor=cyan}
\usepackage{xspace}
\usepackage{thm-restate,color,xcolor}
\usepackage{boxedminipage}
\usepackage[boxed]{algorithm}
\usepackage{epigraph}
\usepackage{amsmath}
\usepackage{mathtools}
\usepackage{framed}
\usepackage[framemethod=tikz]{mdframed}
\usepackage{titlesec}
\usepackage{lipsum}%
\usepackage{cleveref,aliascnt}
\usepackage{tikz}
\usepackage{wrapfig}
\usepackage{framed}
\usepackage[framemethod=tikz]{mdframed} 

\usetikzlibrary{shapes.geometric}
\usetikzlibrary{arrows}
\usetikzlibrary{arrows.meta}
\usetikzlibrary{patterns}
\usetikzlibrary{shapes.misc}

\newcommand{\dist}{\mbox{\rm dist}}

\newtheorem{theorem}{Theorem}[section]
\newtheorem{lemma}[theorem]{Lemma}

\newtheorem{corollary}[theorem]{Corollary}
\newtheorem{claim}[theorem]{Claim}
\newtheorem{observation}[theorem]{Observation}
\newtheorem{question}{Question}

\theoremstyle{definition}
\newtheorem{definition}[theorem]{Definition}

\usepackage{wrapfig}

\usepackage{changepage}

\def\blackslug
     {\hbox{\hskip 1pt\vrule width 8pt height 8pt depth 1.5pt\hskip 1pt}}
\def\qed{\quad\blackslug\lower 8.5pt\null\par}

\newcommand{\poly}{{\rm poly}}

\newcommand{\TCLargeDeg}{\mathsf{TCLargeDeg}}
\newcommand{\ShortcutLargeDeg}{\mathsf{ShortcutLargeDeg}}

\newcommand{\DeltaIn}{\mathsf{\Delta}_{in}}
\newcommand{\DeltaOut}{\mathsf{\Delta}_{out}}
\newcommand{\DeltaInInd}{\mathsf{\Delta}^{in}}
\newcommand{\DeltaOutInd}{\mathsf{\Delta}^{out}}

\newcommand{\NaiveShortcut}{\mathsf{NaiveShortcut}}

\newcommand{\ShortcutTwoPaths}{\mathsf{Shortcut2LenPaths}}

\newcommand{\dShortcut}{\mathsf{dShortcut}}
\newcommand{\ThreeShortcut}{\mathsf{3Shortcut}}

\newcommand{\LargedShortcut}{\mathsf{LargedShortcut}}

\newcommand{\TCNaiveComb}{\mathsf{BoundedDegTC}}
\newcommand{\TCInDeg}{\mathsf{Tdeg}_{in}}
\newcommand{\TCOutDeg}{\mathsf{Tdeg}_{out}}
\newcommand{\TCDeg}{\mathsf{Tdeg}}
\newcommand{\ShortcutInBalls}{\mathsf{ShortcutInBalls}}

\newcommand{\TCHittingSet}{\mathsf{TCHittingSet}}
\newcommand{\TCHittingSetd}{\mathsf{TCHittingSetBalls}}

\newcommand{\InPred}[2]      {\mathsf{In}\ifblank{#1}{(#2)}{(#2,#1)}}
\newcommand{\OutSucc}[2]      {\mathsf{Out}\ifblank{#1}{(#2)}{(#2,#1)}}
\newcommand{\InPredBall}[3]      {\mathsf{In}\ifblank{#1}{(#2,#3)}{(#2,#1,#3)}}
\newcommand{\OutSuccBall}[3]      {\mathsf{Out}\ifblank{#1}{(#2,#3)}{(#2,#1,#3)}}

\usepackage{enumitem}

\usepackage{subcaption} 
\usepackage{booktabs}

\newcommand{\BDS}{\begin{description}}
\newcommand{\EDS}{\end{description}}
\newcommand{\BE}{\begin{enumerate}}
\newcommand{\EE}{\end{enumerate}}
\newcommand{\BI}{\begin{itemize}}
\newcommand{\EI}{\end{itemize}}
\newcommand{\BPF}{\begin{proof}}
\newcommand{\EPF}{\end{proof}}

\newcommand{\BB}{\begin{enumerate}}
\newcommand{\EB}{\end{enumerate}}

\usepackage{booktabs}
\usepackage{tabularx}

\newenvironment{wrapper}[1]
{
	\begin{center}
		\begin{minipage}{\linewidth}
			\begin{mdframed}[hidealllines=true, backgroundcolor=gray!20, leftmargin=0cm,innerleftmargin=0.5cm,innerrightmargin=0.5cm,innertopmargin=0.5cm,innerbottommargin=0.5cm,roundcorner=10pt]
				#1}
			{\end{mdframed}
		\end{minipage}
	\end{center}
} 

\usepackage{geometry}
\title{$\tilde{O}(1)$-Depth Parallel Reachability Faster than Transitive Closure}
\author{
Shimon Kogan\\
        \small Weizmann Institute\\
        \small shimon.kogan@weizmann.ac.il
\and
Merav Parter \thanks{This project is funded by the European Research Council (ERC) under the European Union Horizon 2020 research and innovation programme (grant agreement No. 949083).}\\
        \small Weizmann Institute\\
        \small merav.parter@weizmann.ac.il
}
\date{}

\begin{document}
\maketitle

\begin{abstract}
A $d$-shortcut of a directed graph $G=(V,E)$ is a subset of edges drawn from the transitive closure $TC(G)$ whose addition reduces the graph diameter to at most $d$. In the special case $d=1$, computing a $1$-shortcut is \emph{equivalent} to computing the transitive closure. For larger values of $d$, a lower bound of Hesse~\cite{Hesse03} shows that $n^{\delta}$-shortcuts, for small constants $\delta>0$, may still contain a large fraction of the edges of $TC(G)$, suggesting that shortcut construction may remain as hard as transitive closure even in this regime. Consequently, since $\widetilde{O}(d)$-depth parallel reachability algorithms rely on computing $d$-shortcuts, achieving $\widetilde{O}(1)$ depth by this approach has so far required computing the full transitive closure. Assuming $\omega=2$, the PS-AE-Triangle hypothesis of Abboud, Bringmann, Fischer, and K\"unnemann~\cite{AbboudBFK24}, together with the standard three-layer reduction from sparse Boolean matrix multiplication, yields a conditional $T^{4/3-o(1)}$ time barrier for computing transitive closure when $T\leq n^{3/2}$, where $T=|TC(G)|$.

In this work, we bypass the transitive-closure barrier for $\widetilde{O}(1)$-depth parallel reachability. We introduce randomized $d$-shortcut constructions that already circumvent this barrier for $d=3$ and, more generally, for every even $d\geq4$ up to $O(\log n)$. Our approach yields a randomized $\widetilde{O}(1)$-depth parallel reachability algorithm with total work $\widetilde{O}(T^{\omega/2})$, which becomes $\widetilde{O}(T)$ when $\omega=2$, falling below this conditional $T^{4/3-o(1)}$ barrier throughout that regime. Under the current bound $\omega<2.371339$~\cite{AlmanDVXXXZ25}, this corresponds to $\widetilde{O}(T^{1.186})$ work, improving on the current $T^{1.3459+o(1)}$ sequential-time bound for transitive closure due to Abboud et al.~\cite{AbboudBFK24}. Thus, although $\widetilde{O}(1)$-shortcuts might be almost as dense as the full transitive closure, they can nevertheless be computed substantially faster.

%
%
%
%
\end{abstract}

\newpage

\pagenumbering{gobble}

{\small\tableofcontents}
\newpage

\pagenumbering{arabic}

\section{Introduction}

The well-known \emph{transitive closure bottleneck}~\cite{KaoK90,BernReviving26} in parallel computing refers to the necessity of computing the full transitive closure, typically via iterated matrix squaring, to solve single-source reachability in polylogarithmic depth. In an effort to bypass this barrier, Thorup~\cite{Thorup92} introduced the graph-theoretic notion of \emph{$d$-shortcuts}: a subset of edges drawn from the transitive closure of a graph $G$ whose addition reduces its diameter to at most $d$. His well-known conjecture, stating that every graph admits a $\poly\log n$-shortcut of near-linear size, was later refuted by Hesse~\cite{Hesse03}; nevertheless, it has driven much of the subsequent work in the area and remains a central focus of recent research~\cite{BernReviving26}.

Since their introduction, shortcuts have played a central role in parallel reachability algorithms. Given a $d$-shortcut $H$ of $G$, single-source reachability (SSR) can be solved in parallel with $\widetilde O(d)$ depth and $\widetilde O(m+|H|)$ work. Prior work has primarily focused on minimizing the depth achievable with near-linear work. A breakthrough result of Fineman~\cite{Fineman18} achieved near-linear work with sublinear depth $\widetilde O(n^{2/3})$, later improved to $n^{1/2+o(1)}$ by Liu, Jambulapati, and Sidford~\cite{LiuJS19}. Very recently, Ashvinkumar, Bernstein, Probst Gutenberg, and Saranurak~\cite{AshvinkumarBGS26} gave an SSR algorithm for dense digraphs with $\widetilde O(n^2)$ work and $n^{0.136}$ depth.

A complementary line of research studies the combinatorial properties of shortcuts, focusing mainly on pinpointing the smallest $d$ for which linear-size shortcuts exist. In a sense, this parameter sets a limit on the best possible depth one can hope for in the parallel setting. The current upper and lower bounds are $d=\widetilde O(n^{1/3})$ by Kogan and Parter~\cite{KoganPSODA22} and $d=\widetilde\Omega(n^{1/4})$ by Bodwin and Hoppenworth~\cite{bodwin2023folklore}. A number of other recent works study shortcuts from algorithmic and constructive~\cite{HaeuplerRZ26}, approximation~\cite{ChalermsookJMN26}, and lower-bound~\cite{VassilevskaWXX24} perspectives. Very recently, Bernstein et al.~\cite{BernReviving26} revisited Thorup's conjecture under a variant allowing shortcuts to include Steiner vertices. Despite all efforts, to date it remains unclear how to solve single-source reachability in polylogarithmic depth without computing the entire transitive closure.

In this work, we take a further step toward bypassing this long-standing transitive closure bottleneck for polylogarithmic-depth parallel reachability. In contrast to prior work that insists on near-linear work, we focus on the following fundamental and largely unexplored question:

\vspace{-2pt}
\begin{question}\label{q:low-depth-faster-TC}
Can single-source reachability be solved in polylogarithmic depth with \textbf{asymptotically less work} than is required to compute the entire transitive closure?
\end{question}
\vspace{-2pt}

It is straightforward to see that $1$-shortcuts are computationally equivalent to the transitive closure. However, this equivalence becomes unclear already for $d\geq2$. The seminal \emph{size} lower bound of Hesse~\cite{Hesse03} showed that there are sparse $n$-vertex digraphs for which every $n^\delta$-shortcut must be almost quadratic in size, for sufficiently small constants $\delta>0$. As Hesse remarks:

\begin{quote}
\emph{These lower bounds do imply that for an algorithm in this model to achieve polylogarithmic running time, it must precompute essentially the entire transitive closure.} \hfill---Hesse~\cite{Hesse03}
\end{quote}

This seems to provide evidence against \Cref{q:low-depth-faster-TC}, suggesting that $d$-shortcuts may be computationally as hard as the transitive closure even for $d=n^\delta$, where $\delta>0$ is a sufficiently small constant. Indeed, since such shortcuts offer no asymptotic size advantage over the trivial $1$-shortcut $TC(G)$, it is natural to suspect that they cannot be constructed faster either. Our results show, however, that this size-based intuition is incomplete: an $O(\log n)$-shortcut can be constructed directly, without first computing the entire transitive closure, even on graphs where every such shortcut is nearly as large as the full transitive closure.

Our main algorithmic result is the following. Here $T=|TC(G)|$, and $\omega$ denotes the matrix-multiplication exponent, currently known to satisfy $\omega<2.371339$~\cite{AlmanDVXXXZ25}. All our algorithms are randomized, and their correctness guarantees hold with high probability.

\begin{theorem}[Parallel reachability faster than TC]\label{thm:parallel-ssr}
There is a randomized parallel algorithm that solves SSR using $\widetilde O(T^{\omega/2})$ work and polylogarithmic depth.
\end{theorem}

To place this result in context, consider the sequential complexity of computing the transitive closure. Abboud, Bringmann, Fischer, and K\"unnemann formulated the PS-AE-Triangle hypothesis and used it to obtain conditional lower bounds for sparse Boolean matrix multiplication~\cite[Theorem~1.10]{AbboudBFK24}. Through the standard three-layer reduction from Boolean matrix multiplication to transitive closure, their hypothesis yields the following output-sensitive conditional time barrier when $\omega=2$:
\[
 \mathsf{B}(n,T):=\min\{T^{4/3},n^2\}.
\]
For every constant $\varepsilon>0$ and every $n\leq T\leq n^2$, there is a family of $n$-vertex digraphs with $\Theta(T)$ closure edges on which an $O(\mathsf{B}(n,T)^{1-\varepsilon})$-time transitive-closure algorithm would refute the hypothesis. For $T\leq n^{3/2}$, this gives a $T^{4/3}$ time barrier; for $T\geq n^{3/2}$, a parameterized version of the same three-layer construction gives the $n^2$ barrier. Abboud et al.~\cite{AbboudBFK24} also give the current $T^{1.3459+o(1)}$-time transitive-closure algorithm.

Under the current bound on $\omega$, the work in \Cref{thm:parallel-ssr} is $\widetilde O(T^{1.186})$, below the current $T^{1.3459+o(1)}$ transitive-closure time. More conceptually, when $\omega=2$, \Cref{thm:parallel-ssr} uses $\widetilde O(T)$ work, whereas the ABFK conditional time barrier is superlinear in $T$ whenever $T=o(n^2)$. Thus, relative to this barrier, \Cref{thm:parallel-ssr} answers \Cref{q:low-depth-faster-TC} in the affirmative throughout the subquadratic-output regime. More precisely, for every fixed $\delta\in(0,1)$ and $T\leq n^{2-\delta}$, the conditional lower bound is at least $T^{1+\Omega_\delta(1)}$, giving a polynomial separation.

We also exhibit a combined hard family on which shortcut construction and transitive-closure computation separate conditionally. The family combines a Hesse shortcut-hard component with an ABFK conditional TC-hard component; its role is to show that near-maximal shortcut size does not preclude faster shortcut construction. Assuming $\omega=2$ and the PS-AE-Triangle hypothesis of~\cite{AbboudBFK24}, computing the entire transitive closure of a graph in this family requires $\mathsf{B}(n,T)^{1-o(1)}$ time. In contrast, by \Cref{thm:dshortcut}, an $O(\log n)$-shortcut can be computed in randomized $\widetilde O(T)$ time when $\omega=2$, even though every such shortcut contains $T^{1-o(1)}$ edges.

\begin{theorem}\label{thm:joint-hardness}
For every asymptotic choice $T=T(n)$ with $n\leq T\leq n^2$, there is a family $\mathcal{G}_{n,T}$ of $n$-vertex digraphs such that $|TC(G^*)|=\Theta(T)$ for every $G^*\in\mathcal{G}_{n,T}$. Every $O(\log n)$-shortcut for $G^*$ has $T^{1-o(1)}$ edges. Moreover, assuming the PS-AE-Triangle hypothesis of~\cite{AbboudBFK24}, for every constant $\varepsilon>0$, no algorithm computes the transitive closure of every graph in $\mathcal{G}_{n,T}$ in time $O(\mathsf{B}(n,T)^{1-\varepsilon})$.
\end{theorem}

Thus, although polylogarithmic-diameter shortcuts might be almost as dense as the full transitive closure, they can nevertheless be computed substantially faster. A proof overview appears in \Cref{sec:dheirarchy}, and the complete proof is given in \Cref{app:joint-hardness}.

\smallskip
\noindent\textbf{A Hierarchy of Small-$d$ Shortcut Algorithms.}
Beyond its application to parallel reachability, we seek to understand how the complexity of constructing a $d$-shortcut decreases as the target diameter $d$ grows. While computing a $1$-shortcut is equivalent to computing the transitive closure, the complexity of computing a $2$-shortcut remains open. Our results establish a computational hierarchy beginning at $d=3$ and continuing through every even $d$ up to $O(\log n)$. Formally:

\begin{theorem}\label{thm:dshortcut}
For every $d$ such that either $d=3$ or $d\in[4,O(\log n)]$ is even, a randomized algorithm computes a $d$-shortcut in
\[
 \widetilde O\!\left(T^{\left(1+\frac{1}{d+1}\right)\frac{\omega}{2}}\right)
\]
time. For even $d\geq4$, the algorithm has a parallel implementation with the same work bound and polylogarithmic depth. In particular, when $\omega=2$, the running time is $\widetilde O\!\left(T^{1+\frac{1}{d+1}}\right)$.
\end{theorem}

The $d=3$ guarantee and the even-$d$ hierarchy are obtained by distinct algorithms: \Cref{sec:3-shortcut} gives the specialized construction for $d=3$, whereas \Cref{sec:dheirarchy} gives the hierarchy for even $d\geq4$. Currently, a matching bound is not obtained for odd $d\geq5$.

Assuming $\omega=2$, the hierarchy computes $d$-shortcuts faster than the current transitive-closure algorithm already for $d=3$; under the current bound on $\omega$, this holds for every even $d\geq8$.

\noindent\textbf{Large-Diameter Shortcuts.}
Finally, for superlogarithmic values of $d$, we obtain the following parameterized work--depth tradeoff for constructing $d$-shortcuts. This tradeoff improves the previous square-root-depth bound and gives an output-sensitive work--depth tradeoff for parallel reachability:

\begin{theorem}\label{thm:d-largeshortcutparallel}
For every $d=\Omega(\log^2 n)$, there is a randomized parallel algorithm that computes a $d$-shortcut using
\[
 \widetilde O\!\left(\frac{T^{\omega/2}}{d^{\omega-1}}+m\right)
\]
 work and $\widetilde O(d)$ depth.
\end{theorem}

When $\omega=2$, the shortcut construction uses $\widetilde O(T/d+m)$ work and $\widetilde O(d)$ depth. Under the current bound on $\omega$, setting $d=n^{\frac{\omega-2}{\omega-1}}$ for dense graphs yields $\widetilde O(n^2)$ work and $\widetilde O(n^{0.271})$ depth, improving on the previous $n^{1/2+o(1)}$ depth bound~\cite{LiuJS19,BrandGJV25}. For dense graphs, however, the independently obtained algorithm of~\cite{AshvinkumarBGS26} achieves the stronger $n^{0.136}$-depth bound with $\widetilde O(n^2)$ work. The proof of \Cref{thm:d-largeshortcutparallel} appears in \Cref{sec:large-d-shortcut}.

\subsection{Preliminaries}\label{sec:prelim} \vspace{-5pt}

Our parallel algorithms follow the standard work-depth terminology \cite{blelloch1996programming}: an algorithm's work $W$ is its total number of operations, and its depth $d$ is the length of the longest chain of operations with sequential dependencies. A parallel algorithm is work-efficient if its work bound matches the sequential complexity of the problem.  

\noindent\textbf{Graph Notation.} Unless stated otherwise, all digraphs are unweighted and simple; parallel edges can be discarded. We absorb an additive $n$ into the input-size parameter $m$, equivalently by counting one implicit self-loop at every vertex. Thus $n\leq m\leq |TC(G)|$. 
For an $a$-$b$ path $P_1$ and a $c$-$d$ path $P_2$, let $P_1 \circ (b,c) \circ P_2$ be the $a$-$d$ path obtained by concatenating $P_1$, the connecting edge $(b,c)$, and $P_2$.  
Let $G^{\mathsf{R}}$ denote the graph obtained from $G$ by reversing all edges.
For a given digraph $G=(V,E)$, let $TC(G)$ denote the transitive closure of $G$. We regard every vertex as reachable from itself by the empty path; thus $(v,v)\in TC(G)$ for every $v\in V$.
The set of \emph{successors} of $v$ is
$\OutSucc{G}{v} = \{x \in V(G) ~\mid~ (v,x)\in TC(G)\}$.  We also refer to $\OutSucc{G}{v}$ as the set of vertices that are \emph{reachable} from $v$. The set of \emph{predecessors} of $v$ is
$\InPred{G}{v} = \{x \in V(G) ~\mid~ (x,v)\in TC(G)\}$.  Let $N_{out}(v,G)=\{u ~\mid~ (v,u)\in G\}$ be the outgoing neighbors of $v$ in $G$, and similarly define $N_{in}(v,G)=\{u ~\mid~ (u,v)\in G\}$. Then, $\OutSucc{G}{v}=N_{out}(v,TC(G))$ and $\InPred{G}{v}=N_{in}(v,TC(G))$. The TC-outdegree (resp., indegree) of a vertex $v \in V(G)$ is $\TCOutDeg(v,G)=|\OutSucc{G}{v}|$ (resp., $\TCInDeg(v,G)=|\InPred{G}{v}|$). Define the total TC-degree of $v$ by $\TCDeg(v,G)=\TCInDeg(v,G)+\TCOutDeg(v,G)$, and define the maximum TC-degree of $G$ by $\TCDeg(G)=\max_{v\in V(G)}\TCDeg(v,G)$. The diameter of a digraph $G$, denoted by $D(G)$, is the maximum $u$-to-$v$ distance, if such a pair exists; that is, $D(G)=\max_{(u,v)\in TC(G)}\dist_{G}(u,v)$. A subset $H \subseteq TC(G)$ is a $d$-shortcut for $G$ if the diameter of $G \cup H$ is at most $d$.
\noindent We make use of the following parallel algorithms: 
\begin{theorem}[\cite{Cohen97}]\label{thm:size-estimation}
Given an $n$-vertex digraph $G=(V,E)$, there is a randomized parallel algorithm that, with high probability, computes a $(1+\epsilon)$ estimate of the TC out-degrees $\{\widetilde{d}^{out}(v)\}_{v \in V}$ using $\tilde{O}_{\epsilon}(m)$ work and $\tilde{O}(D(G))$ depth via $O(\log n)$ single-source reachability computations. In particular, $\widetilde{d}^{out}(v) \in (1\pm \epsilon)|\OutSucc{G}{v}|$.
\end{theorem}
We apply this theorem with a fixed constant $\epsilon<1$ to both $G$ and $G^{\mathsf R}$. By scaling the estimates by a constant, we obtain, with high probability, simultaneous upper estimates of every TC-indegree and TC-outdegree, each within a constant factor of its true value. Summing the outdegree estimates similarly gives an upper estimate $T$ satisfying $|TC(G)|\leq T\leq O(|TC(G)|)$. Whenever an algorithm below uses estimated degrees or an estimate $T$, we use these scaled upper estimates.


\begin{theorem}[Parallel Matrix Multiplication, Thm. 9.2 \cite{HuangP98}]\label{thm:parrecmat}
Given an $n \times n$ matrix $A$ and an $n \times n$ matrix $B$, there is a deterministic parallel algorithm for computing the matrix product $A \times B$ using $n^{\omega}$ work and $O(\log n)$ depth.
\end{theorem}

\subsection{Our Approach}

Our approach is based on iteratively shortcutting paths involving high TC-degree vertices. Towards the goal, we introduce the notion of \emph{balanced TC-hitting sets}. Informally, such a set $S$ has bounded total TC-degree while intersecting every sufficiently large predecessor and successor set in $G$.

\begin{definition}[Balanced TC-Hitting Sets]\label{def:TChitset}
For a given $n$-vertex graph $G=(V,E)$ with $T=|TC(G)|$ and an integer $\Delta$, a subset $S \subseteq V$ is a $\Delta$-\emph{balanced TC-hitting set} ($\Delta$-hitting set, in short) if:
\begin{enumerate}
\item $S$ is a hitting set of all successor and predecessor lists of size at least $\Delta$; that is, for every $v$ with $\TCOutDeg(v)\geq \Delta$, $S \cap \OutSucc{G}{v}\neq \emptyset$ and for every $v$ with $\TCInDeg(v)\geq \Delta$, $S \cap \InPred{G}{v}\neq \emptyset$.

\item $\sum_{s \in S} \left(\TCInDeg(s)+\TCOutDeg(s)\right)=\tilde{O}(T /\Delta)$.
\end{enumerate}
\end{definition}

By a standard probabilistic argument, $\Delta$-balanced TC-hitting sets exist and can be computed in randomized near-linear time using the size-estimation technique of \Cref{thm:size-estimation}.
\vspace{-10pt}

\subsubsection{$O(\log n)$-Shortcuts} \label{sec:small-d-overview}

The $O(\log n)$-shortcut algorithm is based on $\ell=4\log n$ iterations, each computing two $1$-shortcut sets, $H_i^{in}$ and $H_i^{out}$, over a subset of vertices $V_i=V(H_i^{in}\cup H_i^{out})$. Let $\Delta_0=n$ and for every $i \in \{1,\ldots, \ell\}$, set $\Delta_i=\sqrt{T}/2^i$; hence $\Delta_{\ell}<1$. In each iteration $i \geq 1$, we are given a graph $G_i \subseteq G$ with maximum TC-degree at most $2\Delta_{i-1}$, where initially $G_1=G$. The goal of iteration $i$ is to shortcut every path having an endpoint $v$ with $\TCDeg(v,G_i)\geq 2\Delta_i$. Such a vertex necessarily satisfies either $\TCOutDeg(v,G_i)\geq \Delta_i$ or $\TCInDeg(v,G_i)\geq \Delta_i$. This is done in two major steps.
First, a $\Delta_i$-hitting set $S_i$ is computed for the graph $G_i$. The second step iterates over $s \in S_i$ and does the following: compute the predecessor and successor lists $\OutSucc{G_i}{s}$ and $\InPred{G_i}{s}$ by computing the incoming and outgoing BFS trees rooted at $s$ in $G_i$. Then, compute the transitive closure of the graphs $G_{out}(s)=G_i[\OutSucc{G_i}{s}]$ and $G_{in}(s)=G_i[\InPred{G_i}{s}]$, by using the naive iterative matrix multiplication algorithm. The output of the iteration is given by:
$$H^{out}_i=\bigcup_{s \in S_i}TC(G_{out}(s)), ~H^{in}_i=\bigcup_{s \in S_i}TC(G_{in}(s)) \mbox{~and~} G_{i+1}=G_i \setminus V(H^{out}_i \cup H^{in}_i)~.$$
The final output shortcut is $H=\bigcup_{i=1}^{\ell}  H^{out}_i \cup H^{in}_i$. This completes the description of the algorithm. 

\smallskip
\noindent \textbf{The Time Argument.} Consider the first iteration. Since $\Delta_1=\Theta(\sqrt{T})$, it holds by \Cref{def:TChitset} that $\sum_{s \in S_1} \left(\TCInDeg(s)+\TCOutDeg(s)\right)=\widetilde O(\sqrt{T})$.
We focus on the time to compute $H^{in}_i$; the argument for $H^{out}_i$ is symmetric. Computing $\InPred{G_i}{s}$ takes $O((\TCInDeg(s))^2)$ time and computing the transitive closure of $G_{in}(s)$ takes $O((\TCInDeg(s))^\omega)$ time, as $|V(G_{in}(s))|=\TCInDeg(s)$. Summing over all $s \in S_1$, the runtime is bounded by:
$$\sum_{s \in S_1}(\TCInDeg(s))^\omega\leq \widetilde O\left((\sqrt{T})^{\omega-1}\cdot \sum_{s \in S_1}\TCInDeg(s)\right)=\widetilde O(T^{\omega/2}).$$

Next consider iteration $i \geq 2$. A critical observation is that the maximum TC-degree of $G_i$ is at most $2\Delta_{i-1}$. For $i=1$ this is immediate. Consider iteration $i$ and observe that for every vertex $v$ with $\TCOutDeg(v,G_i)\geq \Delta_i$, there must be $s \in S_i$ such that $v \in \InPred{G_i}{s}$. Similarly, for every vertex $v$ with $\TCInDeg(v,G_i)\geq \Delta_i$, there must be $s \in S_i$ such that $v \in \OutSucc{G_i}{s}$. Hence, every vertex $v$ with TC-degree at least $2\Delta_{i}$ must appear in $V(H^{in}_i) \cup V(H^{out}_i)$. Consequently, the maximum TC-degree of $G_{i+1}$ is at most $2\Delta_i$. Using also the fact that $\sum_{s \in S_i} \TCInDeg(s)=\widetilde O(T/\Delta_i)$ by \Cref{def:TChitset}, the runtime for computing $H^{in}_i$ is bounded by
$$\sum_{s \in S_i}(\TCInDeg(s))^\omega\leq (2\Delta_{i-1})^{\omega-1}\cdot \sum_{s \in S_i}\TCInDeg(s)=\widetilde O\!\left(\Delta_{i-1}^{\omega-1}\cdot \frac{T}{\Delta_i}\right)=\widetilde O(T^{\omega/2}).$$ This completes the time-bound argument.

\smallskip
\noindent \textbf{The Diameter Argument.} Consider some $u$-$v$ path $P$ in $G$. Let $u'_1$ be the last vertex on $P$ (closest to $v$) that belongs to $V(H^{in}_1)$; if no such vertex exists, the corresponding prefix is empty and is simply omitted. We claim that $\dist_{H^{in}_1}(u,u'_1)=1$. To see this, let $s \in S_1$ be such that $u'_1 \in \InPred{G}{s}$. Since $u'_1\in V(H^{in}_1)$, such $s$ must exist. It then also holds that every $z \in P[u,u'_1]$ belongs to $\InPred{G}{s}$, and consequently 
$P[u,u'_1]\subseteq G_{in}(s)$. Since $H^{in}_1$ contains the transitive closure of $G_{in}(s)$, we have that $\dist_{H^{in}_1}(u,u'_1)=1$, as desired. In a similar manner, let $v'_1$ be the first vertex on $P$ (closest to $u$) that belongs to $V(H^{out}_1)$, omitting the corresponding suffix if no such vertex exists. By a symmetric argument, there exists $s' \in S_1$ such that $P[v'_1,v]\subseteq G_{out}(s')$ and $\dist_{H^{out}_1}(v'_1,v)=1$. Let $P_1=P \setminus V(H^{in}_1 \cup H^{out}_1)$. Observe that $P_1$ is a \emph{connected} subpath of $P$ and is contained in $G_2$. Let $u_1,v_1$ be the endpoints of $P_1$. By repeating the same argument, we can identify a prefix $P_1[u_1,u'_2]$ and a suffix $P_1[v'_2,v_1]$ that are shortcut to length $1$ in $H^{in}_2 \cup H^{out}_2$. Overall, the path $P$ can be partitioned into $k\leq 2\ell=O(\log n)$ segments $P_i$, where each $P_i$ has endpoints $u_i,v_i$ such that $P=P_1 \circ (v_1,u_2) \circ \ldots \circ (v_{k-1},u_{k}) \circ P_k$ and $\dist_{H}(u_i,v_i)=1$. This completes the sketch of the diameter argument; see Fig.~\ref{fig:diam-intro}.

\begin{figure}[h!]
\begin{center}
\includegraphics[width=\linewidth]{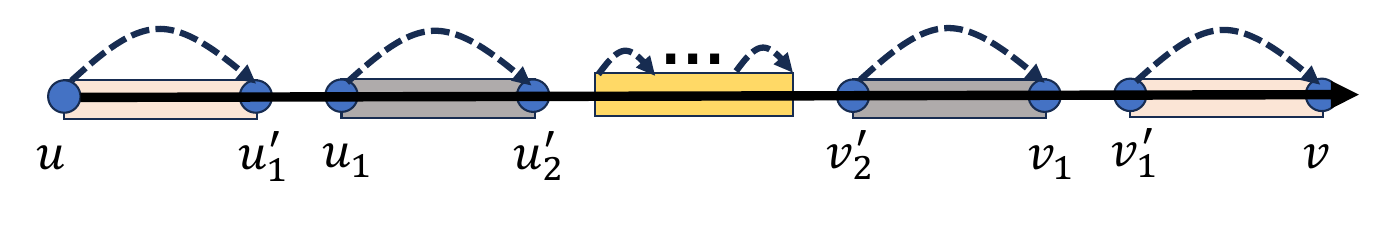}
\caption{Illustration of the $O(\log n)$-diameter argument. Solid black edges belong to the original path, and dashed blue edges are shortcuts.
\label{fig:diam-intro} 
}
\end{center}
\end{figure}

\smallskip
\noindent \textbf{Parallel Implementation in  $\tilde{O}(1)$ Depth.} For the parallel implementation, our starting observation is that the above-mentioned sequential algorithm can be implemented in the parallel setting using $\tilde{O}(D(G))$ depth and a work bound that nearly matches the sequential time. To obtain the desired $\tilde{O}(1)$ bound, for $d=O(\log n)$ we transform the graph $G$ into a $(d+2)$-layered graph $G'$, hence with diameter $d+1$. Each layer of $G'$ consists of copies of the vertices of $G$, and the edges from layer $i$ to $i+1$ correspond to the edges of $G$. Observe that every length-$(d+1)$ path in $G$ is captured by $G'$. We then apply the parallel $d$-shortcut algorithm to $G'$, which can be implemented in $\tilde{O}(D(G'))=\tilde{O}(1)$ depth. The resulting $d$-shortcut $H'$ has the following property: there is a simple mapping from the edges of $H'$ to a set $H\subseteq TC(G)$ such that every $u$-to-$v$ path of length $d+1$ in $G$ is shortened to length at most $d$ in $G \cup H$. In other words, adding $H$ to $G$ reduces the diameter of $G$ by a factor of roughly $(1-1/d)$. Repeating this procedure for $O(d\log n)=\tilde{O}(1)$ iterations on the current augmented graph obtains the desired $d$-shortcut.

We note that \Cref{thm:dshortcut} calls for a more delicate argument because we insist on the exact diameter bound $d$; its specialized $d=3$ construction requires additional care.

\vspace{-10pt}
\subsubsection{$d$-Shortcuts for $d=\Omega(\log^2 n)$}\label{sec:large-d-overview}
The goal is to compute $d$-shortcuts using $\tilde{O}\!\left(\frac{T^{\omega/2}}{d^{\omega-1}}+m\right)$ work and $\tilde{O}(d)$ depth. Notably, using the current bound on $\omega$, this already breaks the $\sqrt{n}$-depth barrier for SSR computation in dense graphs with $\Theta(n^2)$ edges, yielding an improved depth of $\tilde{O}(n^{0.27})$ and near-linear work.

The construction is more intricate than that of Section~\ref{sec:small-d-overview}, with the main challenge lying in the parallel implementation rather than the sequential design. We therefore focus on highlighting the key obstacles and the ideas used to overcome them. Our central task is a diameter-halving phase: we compute a shortcut set $H$ such that every pair at distance exactly $d$ becomes at most $d/2$ apart in $G\cup H$. Repeating this phase $O(\log n)$ times on the current augmented graph yields the desired diameter bound. For ease of exposition, we assume $\omega = 2$ and that $G$ is a DAG; however, our full construction does not rely on either assumption.

Our starting point is the observation that, to improve the running time from $T$ to $O(T/d)$, one must avoid computing the full transitive closure over the incoming neighborhoods of a sampled set of sources $S$ as we did in Sec. \ref{sec:small-d-overview}. Instead, we compute a weaker $d' = d/\log^2 n$ shortcut, which suffices to reduce the diameter by a constant factor. Concretely, we show that for an $n$-vertex graph, one can compute a $d'$-shortcut in time $O(n^2/d')$, which yields the desired $\tilde{\Omega}(d)$ improvement over the construction in Section~\ref{sec:small-d-overview}.

The sequential algorithm is particularly simple. We compute a topological ordering of $G$ and partition it into $d'/4$ consecutive blocks. For each block, we compute its full transitive closure, and return the union of these closures as the shortcut set. The diameter bound follows immediately. For the running time, note that each block has size $O(n/d')$, so computing its transitive closure takes $O((n/d')^{\omega}) = O((n/d')^2)$ time. Summing over all $d'$ blocks yields a total running time of $O(n^2/d')$. This modification captures the essential idea behind the improved sequential construction. In the remainder, we focus on its parallel implementation, where the main technical challenges arise.

Implementing the algorithm in $\tilde{O}(d)$ depth introduces several intertwined challenges. First, the layering approach used in our $\tilde{O}(1)$-shortcut construction is no longer applicable, as it incurs a quadratic dependence on $d$ in the depth.\footnote{For a layered graph of depth $O(d)$, the algorithm requires $d$ iterations, along with $O(d\log n)$ external repetitions.} Instead, we operate directly on the graph $G$, which raises the challenge of estimating $T$, a fundamental quantity in our approach. Moreover, the $d'$-shortcut construction relies on computing a topological ordering of the DAG. Both estimating $T$ and computing such an ordering can be reduced to solving $\tilde{O}(1)$ instances of SSR, which currently requires $\tilde{O}(\sqrt{n})$ depth by~\cite{LiuJS19}.

To facilitate our discussion, we introduce the following terminology. For an integer $k$, the \emph{$k$-incoming ball} of a vertex $u$ consists of all predecessors of $u$ within distance at most $k$ in $G$. This set can be computed with near-linear work and $\tilde{O}(k)$ depth by constructing an incoming BFS tree of depth $k$ rooted at $u$. The \emph{$k$-incoming TC-degree} of $u$ is defined as the size of its $k$-incoming ball. The notions of $k$-outgoing balls and TC-degrees are defined symmetrically.

\textbf{Idea 1: Restricting to $\tilde{O}(d)$-Balls.} Our framework for estimating TC-incoming and outgoing degrees is parameterized by $k = \tilde{O}(d)$. Let $T_k(G)$ denote the sum, over all vertices, of their $k$-incoming and $k$-outgoing TC-degrees. Using the size-estimation technique of \Cref{thm:size-estimation} from~\cite{Cohen97}, this quantity can be approximated with near-linear work and $\tilde{O}(k)$ depth. Since our goal is to shortcut every path of length $d$ to one of length $d/2$, it suffices to consider balls of radius $\tilde{O}(d)$.

However, this restriction introduces a subtle difficulty. As in our $\tilde{O}(1)$-shortcut algorithm, we iteratively handle vertices with sufficiently large $k$-incoming and $k$-outgoing TC-degrees. In the setting of Section~\ref{sec:small-d-overview}, we relied on the property that if a vertex $z$ on a path $P$ has large TC-outdegree, then all its predecessors on $P$ also have large TC-outdegree. In contrast, in the $k$-restricted setting, where we consider vertices with large $k$-outgoing balls, this monotonicity property no longer holds. We address this issue by dynamically adjusting (and in particular, decreasing) the parameter $k$ throughout the iterations of the algorithm.

\textbf{Idea 2: Directed Low Diameter Decomposition (LDD) \cite{BringmannFHL25,HaeuplerJS26}.} To overcome the challenge of computing a topological ordering in $\tilde{O}(d)$ depth, we employ the powerful low-diameter decomposition (LDD) framework of~\cite{BringmannFHL25,HaeuplerJS26}. In particular, the recent work of Haeupler, Jiang, and Saranurak~\cite{HaeuplerJS26} demonstrates how LDD can be leveraged to simplify parallel shortcut constructions. In our setting, LDD plays a crucial role in achieving the desired depth bound. Concretely, it yields a collection of subgraphs $G_1,\ldots,G_\ell$ with $\ell = O(\log n)$, where each $G_i$ is equipped with a topological ordering, and satisfies the following property: for every pair of vertices $u,v$ at distance $d$ in $G$, there exists an index $i$ such that $\dist_{G_i}(u,v) = \dist_G(u,v)$. This collection can be constructed using near-linear work and $\tilde{O}(d)$ depth. We then apply our shortcut procedure independently to each of these $O(\log n)$ subgraphs, thereby capturing all relevant length-$d$ paths.

The final algorithm combines the above ideas in $O(\log n)$ degree-threshold iterations. In every iteration, we apply the bounded-neighborhood shortcut procedure in both the incoming and outgoing directions, halve the degree threshold, and decrease the neighborhood radius by $d$. Along any shortest path of length $d$, the procedure shortcuts a prefix and a suffix and leaves a residual path in a lower-degree induced subgraph. After $O(\log n)$ iterations, the residual path is short enough that the entire path has length at most $d/2$. Repeating this phase on the current augmented graph produces the desired $d$-shortcut. The complete algorithm and its analysis appear in \Cref{sec:large-d-shortcut}.
%

\section{$d$-Shortcuts for $d=O(\log n)$}\label{sec:dheirarchy}

We start with the following lower bound of Hesse~\cite{Hesse03}, which supplies the shortcut-hard component in the proof of \Cref{thm:joint-hardness}.

\begin{theorem}[\cite{Hesse03,AbboudBP18}]\label{thm:hesse}
For every sufficiently small constant $\epsilon>0$ and every sufficiently large $n$, there exists a directed $n$-vertex graph with $n^{1+\epsilon}$ edges and diameter $n^{\delta}$, for $\delta=\Theta(\epsilon^3)$, that requires a shortcut set of $\Omega(n^{2-\epsilon})$ edges to reduce the diameter to $o(n^{\delta})$.
\end{theorem}

\paragraph{Proof overview of \Cref{thm:joint-hardness}.}
The construction combines three components. A Hesse graph from \Cref{thm:hesse} forces every $O(\log n)$-shortcut to contain $T^{1-o(1)}$ edges, an ABFK hard instance gives the conditional $\mathsf{B}(n,T)^{1-o(1)}$ transitive-closure lower bound, and a directed path adjusts the closure size to $\Theta(T)$. The components are joined only through outgoing edges from a new source, preserving both hardness properties without creating many additional closure edges. See \Cref{app:joint-hardness} for the complete proof, including the choice of $\epsilon$.

In the following, we first present the sequential algorithms underlying this hierarchy and then give the parallel implementation that yields our key result, \Cref{thm:dshortcut}.

\vspace{-10pt}
\subsection{Sequential Construction}

At a high level, the algorithm performs $O(d)$ \emph{peeling steps}, one for each degree threshold. In step $i$, it finds a predecessor-closed set $V_i$ containing every vertex whose TC-outdegree in the current graph $G_i$ meets the threshold, adds the induced closure $TC(G_i)[V_i]$ to the shortcut, and continues on the reversed residual graph $(G_i\setminus V_i)^{\mathsf R}$. Along any directed path in $G_i$, the vertices of $V_i$ form a prefix, which the induced closure compresses to one edge. Reversing the residual graph converts its TC-outdegree bound into the TC-indegree bound needed at the next step. The alternating reversals leave one final boundary between path blocks, which is handled by a separate set of stitching edges.

The sets $V_i$ are found using \emph{balanced TC-hitting sets}; see \Cref{def:TChitset}. Informally, such a set $S$ intersects every sufficiently large predecessor and successor set while having bounded total transitive-closure degree.

\begin{claim}\label{cl:computeTC-hit-set}
There is a randomized near-linear algorithm $\TCHittingSet$ that given an $n$-vertex graph $G$ and parameter $\Delta$, w.h.p. outputs a $\Delta$-balanced TC-hitting set $S \subseteq V(G)$ in near-linear time. The algorithm can be implemented in parallel with the same time bound and $\tilde{O}(D(G))$ depth.
\end{claim}
\begin{proof}
Compute $k=\Theta(\log n)$ vertex subsets $S_1,\ldots,S_k$, where each $S_j$ is obtained by sampling every vertex in $V(G)$ independently with probability $p=\min\{1,c\log n/\Delta\}$, for a sufficiently large constant $c$. For a fixed candidate and a fixed successor or predecessor list of size at least $\Delta$, the probability of missing the list is at most $(1-p)^\Delta\leq e^{-p\Delta}\leq n^{-c}$ when $p<1$, and is zero when $p=1$. A union bound over the $k$ candidates and the at most $2n$ lists shows that, with high probability, every candidate simultaneously satisfies property~(1) of \Cref{def:TChitset}. We next show that property~(2) holds with constant probability for each $S_j$. Consequently, with high probability at least one of the $k$ sets satisfies both properties.
In expectation over the sampling of $S_j$, we have $\sum_{s \in S_j} \left(\TCInDeg(s)+\TCOutDeg(s)\right)=p \cdot 2T=2p\cdot T$. 
By Markov's inequality, the probability that $\sum_{s \in S_j} \left(\TCInDeg(s)+\TCOutDeg(s)\right)\geq 4p\cdot T$ is at most $1/2$. Therefore, with high probability there exists a set $S_j$ satisfying property~(2). Apply \Cref{thm:size-estimation} to both $G$ and $G^{\mathsf R}$ to obtain constant-factor estimates of the TC-outdegrees and TC-indegrees, respectively. These estimates allow us to select a set $S_j$ whose total estimated TC-degree is minimum; with high probability, its actual total TC-degree is $O(pT)=\widetilde O(T/\Delta)$. The running time is near-linear in $|E(G)|$.

For the parallel implementation, by \Cref{thm:size-estimation}, it holds that the TC-indegrees and TC-outdegrees estimation can be done in near-linear work and $\tilde{O}(D(G))$ depth, which dominates the time and depth bounds of the algorithm. 
\end{proof}

Our key tool is an algorithm that, given a graph with bounded TC-indegree, computes the transitive-closure edges between all vertices in $V'$ having sufficiently large TC-outdegree. The algorithm does not explicitly compute $V'$; instead, it computes a subgraph of the transitive closure that contains the desired $V' \times V'$ edges.

\begin{theorem}\label{thm:TC-large-deg}
There is an algorithm $\TCLargeDeg$ that, given an $n$-vertex $m$-edge digraph $G$ with $T=|TC(G)|$, maximum TC-indegree at most $\DeltaIn$, and an integer $1\leq\DeltaOut\leq T$, computes a vertex set $U$ and the induced closure $TC'=TC(G)[U]$. The set $U$ contains every vertex whose TC-outdegree is at least $\DeltaOut$; in particular, $TC'$ has diameter at most $1$. The runtime is
$$\tilde{O}\left(m+\min\left\{\frac{T}{\DeltaOut}, \DeltaIn\right\}^{\omega-1} \cdot \frac{T}{\DeltaOut}\right).$$ 
\end{theorem}
The algorithm is used only in this parameter range. If $\DeltaOut>T$, no vertex qualifies, and the empty output handles the degenerate case.

\paragraph{Algorithm $\TCLargeDeg$.} The algorithm has two steps. The first step applies Alg. $\TCHittingSet$ of \Cref{cl:computeTC-hit-set} to compute a $\DeltaOut$-TC hitting set $S$. The second step computes the transitive closure edges between each pair of vertices in $\InPred{G}{s}$ for each $s \in S$, namely, it computes the transitive closure of the graphs $G[\InPred{}{s}]$ for every $s \in S$. 
This is done as follows. For every $s\in S$, compute a directed incoming BFS tree rooted at $s$ in $G$, thereby obtaining $\InPred{}{s}$. Apply the naive matrix-multiplication-based TC computation to $G[\InPred{}{s}]$, and return $TC'=\bigcup_{s \in S} TC(G[\InPred{}{s}])$ and $U=V(TC')$.
This completes the description of the algorithm. 

\begin{proof}[Proof of \Cref{thm:TC-large-deg}]
\noindent \textbf{Correctness.} By reflexivity, $U=\bigcup_{s\in S}\InPred{G}{s}$. We first prove that $TC'=TC(G)[U]$. The containment $TC'\subseteq TC(G)[U]$ is immediate. Conversely, consider $u,v\in U$ with $(u,v)\in TC(G)$. Choose $s\in S$ such that $v\in\InPred{G}{s}$. Every vertex on a $u$-to-$v$ path reaches $v$ and then $s$, so the entire path lies in $G[\InPred{G}{s}]$. Hence $(u,v)\in TC(G[\InPred{G}{s}])\subseteq TC'$. This proves the equality and, in particular, $D(TC')\leq1$.

It remains to show that $U$ contains every vertex $v$ with $\TCOutDeg(v,G)\geq\DeltaOut$. By \Cref{cl:computeTC-hit-set}, the set $S$ hits $\OutSucc{G}{v}$. Thus some $s^*\in S\cap\OutSucc{G}{v}$ is reachable from $v$, which means $v\in\InPred{G}{s^*}\subseteq U$.

\noindent \textbf{Runtime.} By \Cref{cl:computeTC-hit-set}, $S$ is a $\DeltaOut$-TC hitting set and therefore, $\sum_{s \in S}\TCInDeg(s)=\tilde{O}(T /\DeltaOut)$. 
The set $\InPred{G}{s}$ can be computed in $O(|\InPred{G}{s}|^2)$ time: a reverse BFS scans an edge $(x,y)$ only after reaching $y$, and then both $x$ and $y$ belong to $\InPred{G}{s}$. Since the graph is simple, there are at most $|\InPred{G}{s}|^2$ such edges. Computing $TC(G[\InPred{}{s}])$ takes $(\TCInDeg(s))^\omega$ time. Hence, the total time for these $|S|$ transitive-closure computations is

\begin{eqnarray*}
\sum_s (\TCInDeg(s))^{\omega} &\leq& \min\{\sum_{s\in S} \TCInDeg(s), \DeltaIn\}^{\omega-1} \cdot \sum_s \TCInDeg(s)
\\&\leq& \tilde{O}\left(\min\left\{\frac{T}{\DeltaOut}, \DeltaIn\right\}^{\omega-1}\cdot \frac{T}{\DeltaOut}\right)~.
\end{eqnarray*}

\end{proof}

The following lemma will be useful in our arguments.

\begin{lemma}\label{lem:close-P}
Let $TC'$ be the output of algorithm $\TCLargeDeg$ for some input values of $\DeltaIn$ and $\DeltaOut$. Then, for any path $P$, the graph $P \setminus V(TC')$ is a possibly empty contiguous subpath of $P$. 
\end{lemma}
\begin{proof}
Let $S$ be the $\DeltaOut$-TC hitting set for $G$ and recall that $TC'=\bigcup_{s \in S} TC(G[\InPred{G}{s}])$. Fix a $u$-$v$ path $P$ in $G$. If $P\cap V(TC')=\emptyset$, the claim is immediate. Otherwise, let $u'$ be the last vertex on $P$ (closest to $v$) that belongs to $V(TC')$. We show that every vertex $w$ on the segment $P[u,u']$ belongs to $V(TC')$. Choose $s\in S$ such that $u'\in\InPred{G}{s}$. Since $w$ reaches $u'$ and $u'$ reaches $s$, we have $w\in\InPred{G}{s}\subseteq V(TC')$. Thus deleting $V(TC')$ leaves precisely a suffix of $P$, possibly empty.
\end{proof}

We also need the following primitive: 
\begin{lemma}\label{lem:naive-comb-small-deg}
There is an algorithm $\TCNaiveComb$ that, given a graph $G$ with $T=|TC(G)|$ and maximum TC-indegree at most $\Delta$, computes $TC(G)$ in time $\tilde{O}(T \cdot \Delta)$. 
\end{lemma}
\begin{proof}
The algorithm simply computes an incoming BFS tree $B_v$ from each vertex $v \in V(G)$ in $G$. The output is given by $TC(G)=\{(u,v) ~\mid~ u \in B_v, v\in V\}$. The correctness is immediate and we bound the runtime. For every vertex $v$ with TC-indegree $\ell$, the computation takes $O(\ell^2)$ time. Partition the vertices into $O(1+\log \Delta)$ classes, where $V_j$ contains the vertices whose TC-indegree lies in $[2^j,2^{j+1})$, for $j=0,\ldots,\lfloor\log_2\Delta\rfloor$. We then have that $|V_j|\leq T/2^j$ and hence computing the incoming BFS trees from each vertex in $V_j$ can be done in time $O(|V_j|\cdot 2^{2j+2})=O(2^j \cdot T)$. Summing over all classes, the total computation time is $\tilde{O}(T \cdot \Delta)$, as required.
\end{proof}

\paragraph{Algorithm $\dShortcut$ (Even $d$).} Throughout, let $T$ be the constant-factor upper estimate of $|TC(G)|$ described after \Cref{thm:size-estimation}.
The algorithm has two phases. The first phase works in $d/2$ iterations, for $i \in \{d/2,\ldots, 1\}$. In every iteration $i$, the input is a graph $G_i$ with maximum TC-indegree at most $\DeltaInInd_i$ and TC-outdegree threshold $\DeltaOutInd_i$, and the output is a graph $G_{i-1}$ whose edge directions are reversed. Initially $G_{d/2}=G$ and $\DeltaInInd_{d/2}=T$, and for any $i \in \{d/2,\ldots, 1\}$: 

\begin{equation}\label{eq:Delta}
\DeltaOutInd_i=T^{i/(d+1)} \mbox{~and~} \DeltaInInd_{i-1} =\DeltaOutInd_{i}.
\end{equation}
Rounding these thresholds to integers affects the bounds only by constant factors.

The algorithm then applies $\TCLargeDeg$ with inputs $G_i$, $\DeltaInInd_i$, and $\DeltaOutInd_i=T^{i/(d+1)}$. Letting $TC_i$ be the output, define $G_{i-1}=(G_i \setminus V(TC_i))^{\mathsf R}$. 

\begin{wrapper}\vspace{-2pt}
\begin{center}\textbf{Algorithm $\dShortcut(G,d)$}
\end{center}
\begin{enumerate}
\item $G_{d/2}=G$.
\item Use \Cref{thm:size-estimation} to compute an estimate $T$ to $|TC(G)|$.
\item For $i \in \{d/2,\ldots, 1\}$ do:
\begin{enumerate}
\item Set $\DeltaInInd_{i}$ and $\DeltaOutInd_i$ as in Eq. (\ref{eq:Delta}).
\item $TC_i=\TCLargeDeg(G_i, \DeltaInInd_{i}, \DeltaOutInd_{i})$. 
\item $G_{i-1}=(G_i\setminus V(TC_i))^{\mathsf{R}}$.
\end{enumerate}
\item $TC_0=\TCNaiveComb(G_{0},\DeltaInInd_{0})$.
\item For every $v \in V(G_1)$: $E_{0}(v)=\{(v,w) ~\mid~ \exists u \in N_{out}(v,G_1)\cap V(G_0) \mbox{~and~} (w,u) \in TC_0\}$. 
\item If $d/2$ is even, $E_{1 \to 0}=\bigcup_{v \in V(G_1)} E^{\mathsf{R}}_{0}(v)$ and otherwise, 
$E_{1 \to 0}=\bigcup_{v \in V(G_1)} E_{0}(v)$.
\item Set $I_{even}=\{i ~\mid~ (d/2-i) \mbox{~is even}\}$ and $I_{odd}=\{0,\ldots, d/2\}-I_{even}$.
\item Return $\bigcup_{i \in I_{even}} TC_i \cup \bigcup_{i \in I_{odd}}(TC_i)^{\mathsf{R}} \cup E_{1 \to 0}$.
\end{enumerate}
\end{wrapper}

\noindent \textbf{Correctness.} We first show that $H$ is a subset of the transitive closure, and then give the diameter argument. Let $V_i=V(TC_i)$ for every $i \in \{0,\ldots, d/2\}$. 

\begin{lemma}\label{lem:containment}
$H \subseteq TC(G)$. 
\end{lemma}
\begin{proof}
Observe that $TC_i \subseteq TC(G_i)$ for every $i \in \{d/2,\ldots,0\}$ and $TC^{\mathsf{R}}_0\subseteq TC(G_1)$. If $(v,w)\in E_0(v)$, then $(v,u)\in E(G_1)$ and $(w,u)\in TC_0$ for some $u\in V(G_0)$. Thus $(u,w)\in TC(G_1)$ and $(v,w)\in TC(G_1)$. It is therefore sufficient to prove that $G_i \subseteq G$ if $i \in I_{even}$ and $G_i^{\mathsf R} \subseteq G$ otherwise. We show this by induction from $i=d/2$ to $0$. For $i=d/2$, $G_i=G$. Assuming the claim for $i$, it also holds for $i-1$ because $G_{i-1}$ is the reversal of a vertex-induced subgraph of $G_i$.
\end{proof}

\begin{observation}\label{cl:max-out-induc}
For every $i \in \{0,\ldots,d/2\}$, the maximum TC-indegree of $G_i$ is at most $\DeltaInInd_i$.
\end{observation}
\begin{proof}
For $i=d/2$, $\DeltaInInd_{d/2}=T$ and the claim holds vacuously. For every $i\in\{1,\ldots,d/2\}$, let $V'_i$ be the set of all vertices with TC-outdegree at least $\DeltaOutInd_i$ in $G_i$. By \Cref{thm:TC-large-deg}, $V'_i\subseteq V(TC_i)$. We then have that the maximum TC-outdegree of the graph $G_i\setminus V(TC_i)$ is at most $\DeltaOutInd_i$. Since $G_{i-1}=(G_i\setminus V(TC_i))^{\mathsf{R}}$, we get that the maximum TC-indegree of $G_{i-1}$ is at most $\DeltaOutInd_i=\DeltaInInd_{i-1}$.
\end{proof}

The next lemma follows by repeated applications of \Cref{lem:close-P}.
\begin{lemma}\label{lem:connectivity}
Let $P$ be a $u$-$v$ path in $G$. Then, for every $i \in \{d/2,\ldots, 1\}$, the graph $Q_i=P \setminus \bigcup_{j=i}^{d/2}V_j$ is a possibly empty contiguous subpath of $P$. Moreover, for $i \in I_{even}$, $Q_i^{\mathsf R}\subseteq P^{\mathsf R}\cap G_{i-1}$, whereas for $i \in I_{odd}$, $Q_i\subseteq P\cap G_{i-1}$. 
\end{lemma}
\begin{proof}
For every $i$, let $Q_i=P \setminus \bigcup_{j=i}^{d/2}V_j$ and recall that $G_{i-1}=(G_i\setminus V_i)^{\mathsf R}$. We prove the claims by induction from $i=d/2$ to $i=1$.

For $i=d/2$, \Cref{lem:close-P} shows that $Q_{d/2}=P \setminus V_{d/2}$ is a possibly empty contiguous subpath of $P$. Reversing it gives $Q_{d/2}^{\mathsf R}\subseteq G_{d/2-1}=(G\setminus V_{d/2})^{\mathsf R}$.
Assume the claim holds for $i$ and consider $i-1$. Suppose first that $i \in I_{even}$; the other case is symmetric. Then $Q_i^{\mathsf R}$ is a path in $G_{i-1}$. Applying \Cref{lem:close-P} to this path shows that $Q_i^{\mathsf R}\setminus V_{i-1}=Q_{i-1}^{\mathsf R}$ is a possibly empty contiguous subpath. Reversing once more gives $Q_{i-1}\subseteq G_{i-2}$, as required.
\end{proof}

\begin{lemma}\label{lem:d-shortcut-diam}
$H$ is a $d$-shortcut.
\end{lemma}
\begin{proof}
Consider a $u$-$v$ path $P$. By Lemma \ref{lem:connectivity}, it holds that $P$ can be decomposed into at most $d/2+1$ subpaths $P_1,\ldots, P_\ell$ where each $P_i$ is a $u_i$ to $v_i$ path and the following holds: (i) $P_1 \circ (v_1,u_2) \circ P_2 \circ \ldots \circ (v_{\ell-1},u_{\ell}) \circ P_\ell=P$ and (ii) $V(P_j)\subseteq V_{i_j}$ for distinct $i_j$ indices in $\{0,\ldots, d/2\}$. The edges $(v_i,u_{i+1})$ are called connecting edges, as they connect a $P_i$ path with $P_{i+1}$. 

For $j\geq1$, \Cref{thm:TC-large-deg} gives $TC_j=TC(G_j)[V_j]$, while $TC_0=TC(G_0)$. Let $TC'_j=TC_j$ if $j \in I_{even}$ and $TC'_j=TC_j^{\mathsf R}$ otherwise. Each block $P_i$ is oriented as a path in $G_{i_j}$ after applying the appropriate reversal, and therefore $(u_i,v_i)\in TC'_{i_j}\subseteq H$. Adding the $\ell-1$ connecting edges $(v_j,u_{j+1})$ between consecutive paths $P_j$ and $P_{j+1}$ gives $\dist_{G \cup H}(u,v)\leq 2\ell-1$. If the residual subpath $Q=P\setminus\bigcup_{j=2}^{d/2}V_j$ meets at most one of $V_1$ and $V_0$, then at least one degree class is absent, so $\ell\leq d/2$ and the desired bound follows. It remains to consider the case where $Q$ meets both $V_1$ and $V_0$. Assume that $d/2$ is odd; the other case is symmetric. Then $G_1 \subseteq G$.

By \Cref{lem:connectivity}, $Q \subseteq G_1$. Moreover, there exists $i \in \{1,\ldots, \ell-1\}$ such that $Q=P_i \circ (v_i,u_{i+1}) \circ P_{i+1}$, where $P_i \subseteq V_1$ and $P_{i+1}\subseteq V_0$. Recall that $(u_i,v_i),(u_{i+1},v_{i+1})\in H$.
We claim next that $(v_i,v_{i+1})\in E_{0}(v_i)$. This holds as $u_{i+1} \in N_{out}(v_i,G_1)$ and $(v_{i+1},u_{i+1})\in TC_0$.
Altogether, this provides a $2$-length path $[u_i,v_{i},v_{i+1}]$ from $u_i$ to $v_{i+1}$ in $G \cup H$. 
The portion outside $Q$ consists of $\ell-2$ shortcut subpaths and $\ell-2$ connecting edges. Therefore,
\[
\dist_{G \cup H}(u,v)
\leq 2(\ell-2)+\dist_{G\cup H}(u_i,v_{i+1})
\leq 2\ell-2\leq d.
\]
See Fig. \ref{fig:diam3} for an illustration for $d=4$. 
\end{proof}

\noindent\textbf{Time Analysis.} The next lemmas bound the computation time of the $TC_i$ graphs and the set $E_{1 \to 0}$, respectively. 

\begin{claim}\label{cl:computeTCi}
$TC_i$ is computed in time $\widetilde O\!\left(T^{\left(1+\frac{1}{d+1}\right)\frac{\omega}{2}}\right)$ for every $i \in \{0,\ldots, d/2\}$.
\end{claim}
\begin{proof}
In every iteration $i\in \{1,\ldots,d/2-1\}$, by \Cref{thm:TC-large-deg}, Alg. $\TCLargeDeg$ with the input $G_i, \DeltaInInd_{i}$ and $\DeltaOutInd_{i}$ is implemented in time
$\widetilde O\!\left(T^{\frac{(i+1)(\omega-1)}{d+1}} \cdot T^{1-\frac{i}{d+1}}\right)\leq \widetilde O\!\left(T^{\left(1+\frac{1}{d+1}\right)\frac{\omega}{2}}\right)$, where the inequality follows from $i\leq d/2$ and $\omega\geq2$. For $i=d/2$, the algorithm is implemented in time
$\widetilde O\!\left(T^{\left(1-\frac{d}{2(d+1)}\right)\omega}\right)=\widetilde O\!\left(T^{\left(1+\frac{1}{d+1}\right)\frac{\omega}{2}}\right)$. 
The additive input-reading cost in \Cref{thm:TC-large-deg} is dominated because $m\leq T$. It remains to bound the computation time of $TC_0$. By \Cref{cl:max-out-induc}, the graph $G_0$ has maximum TC-indegree at most $\DeltaInInd_0=T^{1/(d+1)}$. Thus \Cref{lem:naive-comb-small-deg} gives $\widetilde O(T^{1+1/(d+1)})$ time, which is within the claimed bound because $\omega\geq2$.
\end{proof}

\begin{claim}\label{cl:computeE10}
The computation of the set $E_{1 \to 0}$ takes $O(T^{(1+1/(d+1))})$ time.
\end{claim}
\begin{proof}
Every vertex of $V(G_0)$ has TC-outdegree less than $T^{1/(d+1)}$ in the residual graph $G_1\setminus V(TC_1)$, or equivalently TC-indegree at most $T^{1/(d+1)}$ in $G_0$. Hence
\[
 |E_0(v)|\leq
 \sum_{u\in N_{out}(v,G_1)\cap V(G_0)}\TCInDeg(u,G_0)
 \leq |N_{out}(v,G_1)|\,T^{1/(d+1)}.
\]
Summing over $v\in V(G_1)$ and using $|E(G_1)|\leq T$ proves the claim.
\end{proof}

\subsection{Parallel Construction}
In this section we give a parallel construction of an $O(\log n)$-shortcut using $\widetilde O(T^{\omega/2})$ work and $\widetilde O(1)$ depth, which immediately yields parallel single-source reachability with the same work and depth bounds. We show:

\begin{theorem}\label{thm:dshortcutparallel}
For every even integer $d\in[4,O(\log n)]$ and every $n$-vertex digraph $G$, there is a randomized parallel algorithm for computing a $d$-shortcut using $\widetilde O\!\left(T^{\left(1+\frac{1}{d+1}\right)\frac{\omega}{2}}\right)$ work and $\widetilde O(1)$ depth.
\end{theorem}

\begin{proof}[Proof of \Cref{thm:parallel-ssr}]
Choose an even $d=\Theta(\log n)$ and apply \Cref{thm:dshortcutparallel}. Since $T\leq n^2$, constructing the shortcut uses $\widetilde O(T^{\omega/2})$ work and polylogarithmic depth. Its size is bounded by the same work bound. A standard level-synchronous reachability search in $G\cup H$ then uses $\widetilde O(m+|H|)=\widetilde O(T^{\omega/2})$ work and $\widetilde O(d)=\widetilde O(1)$ depth, because $m\leq T$. The only randomization is in the shortcut construction, whose correctness holds with high probability.
\end{proof}

The argument has two ingredients. We first show that Algorithm $\dShortcut$ can be implemented in $\widetilde O(D(G))$ depth and with a work bound matching the sequential construction. The key point is that each peeled set is predecessor-closed: its deletion preserves all distances between surviving vertices, while reversal preserves diameter. The second part applies this subroutine to diameter-$(d+1)$ layered graphs and makes $\widetilde O(1)$ such applications.

\begin{lemma}\label{lem:layered-shortcut-map}
Let $G^*$ be the $(d+2)$-layered graph obtained from a digraph $G$ by placing a copy of every vertex in each layer and mapping every edge of $G$ between each pair of consecutive layers. If $H^*$ is a $d$-shortcut for $G^*$, then its projection $H$ onto $V(G)\times V(G)$ is contained in $TC(G)$ and shortens every length-$(d+1)$ path in $G$ to length at most $d$.
\end{lemma}
\begin{proof}
Every edge of $H^*$ represents a directed path in $G^*$ and therefore projects to a reachable pair in $G$, so $H\subseteq TC(G)$. A length-$(d+1)$ path in $G$ maps to a path from the first to the last layer of $G^*$. Since $H^*$ is a $d$-shortcut, these two copies are joined by a path of length at most $d$ in $G^*\cup H^*$. Projecting that path back to $G\cup H$ proves the claim.
\end{proof}

\paragraph{Parallel Implementation of Alg. $\dShortcut$.} We next show how to implement Alg. $\dShortcut$ in $\widetilde O(D(G))$ depth.

\begin{lemma}\label{lem:dShortcut-parallel}
\emergencystretch=1em
The parallel implementation of Alg. $\dShortcut$ takes $\widetilde O\!\left(T^{\left(1+\frac{1}{d+1}\right)\frac{\omega}{2}}\right)$ work and $\widetilde O(D(G))$ depth.
\end{lemma}
We make the following observations. 
\begin{observation}\label{obs:TCHittingSet-parallel}
Algorithm $\TCHittingSet$ with input $\Delta$ can be implemented in near-linear work and 
$\widetilde O(D(G))$ depth.
\end{observation}
\begin{proof}
It suffices to show that given a set $S$, one can verify if $\sum_{s \in S} \left(\TCInDeg(s)+\TCOutDeg(s)\right)=\widetilde O(T/\Delta)$ in near-linear work and 
$\widetilde O(D(G))$ depth. We simply use the estimation of \Cref{thm:size-estimation} that obtains a constant approximation of $\TCInDeg(v),\TCOutDeg(v)$ for every $v \in V$ by $O(\log n)$ single-source reachability computations, which takes near-linear work and $\widetilde O(D(G))$ depth.
\end{proof}

\begin{observation}\label{obs:TClargeDeg-parallel}
Algorithm $\TCLargeDeg$ can be implemented with a work bound that matches the sequential time and with $\widetilde O(D(G))$ depth. 
\end{observation}
\begin{proof}
By \Cref{obs:TCHittingSet-parallel}, the computation of the $\DeltaOut$-TC hitting set $S$ takes 
near-linear work and $\widetilde O(D(G))$ depth. The incoming BFS trees rooted at the vertices of $S$ can be computed in $\widetilde O(D(G))$ depth. Using \Cref{thm:parrecmat}, the computation of the transitive closure of $G[\InPred{}{s}]$ can be done in $\widetilde O(1)$ depth and nearly matching sequential time. 
\end{proof}

\begin{lemma}\label{lem:diameter-monotonicity}
Let $F$ be an input graph to $\TCLargeDeg$, let $U$ be the vertex set it returns, and let $R=F\setminus U$. Then, for every $u,v\in V(R)$,
\[
 \dist_R(u,v)=\dist_F(u,v),
\]
and hence $D(R)\leq D(F)$. Consequently, every peeling step of Algorithm $\dShortcut$ satisfies
\[
 D(G_{i-1})=D\bigl((G_i\setminus V(TC_i))^{\mathsf R}\bigr)\leq D(G_i).
\]
\end{lemma}
\begin{proof}
By the construction in \Cref{thm:TC-large-deg}, $U=\bigcup_{s\in S}\InPred{F}{s}$ and is therefore predecessor-closed: if $x\in U$ and $y$ reaches $x$ in $F$, then $y\in U$. Equivalently, $R$ is successor-closed: if a vertex of $R$ could reach a vertex of $U$, predecessor-closure would place that vertex in $U$. Hence every path in $F$ whose first vertex lies in $V(R)$ is entirely contained in $R$, proving that all surviving-pair distances are unchanged. In particular, $D(R)\leq D(F)$. Since reversing all edges preserves the diameter, the claimed inequality follows at every peeling step.
\end{proof}

\begin{proof}[Proof of \Cref{lem:dShortcut-parallel}]
The constant-factor upper estimate $T$ for $|TC(G)|$ can be computed by \Cref{thm:size-estimation} in near-linear work and $\widetilde O(D(G))$ depth. By \Cref{lem:diameter-monotonicity}, $D(G_i)\leq D(G)$ for every intermediate graph $G_i$. Hence, by \Cref{obs:TClargeDeg-parallel}, the computation of each $TC_i$ takes $\widetilde O\!\left(T^{\left(1+\frac{1}{d+1}\right)\frac{\omega}{2}}\right)$ work and $\widetilde O(D(G))$ depth. 
The computation of $TC_0$ is obtained by an immediate parallel implementation of Alg. $\TCNaiveComb$. This algorithm computes an incoming BFS tree from each $v \in V(G_0)$ in $\widetilde O(D(G_0))\leq\widetilde O(D(G))$ depth, and its work bound $\widetilde O(T^{1+1/(d+1)})$ matches the sequential implementation. Finally, the set $E_{1\to0}$ can be computed in parallel using $\widetilde O(T^{1+1/(d+1)})$ work and $\widetilde O(1)$ depth.
\end{proof}

\paragraph{Parallel Computation of a $d$-Shortcut in $\widetilde O(1)$ Depth.} The algorithm has $\ell=O(d\log n)$ steps. Let $G_1=G$. In every step $i \in \{1,\ldots, \ell\}$, given a graph $G_i$, the algorithm transforms it into a diameter-$(d+1)$ graph $G^*_i$ and applies the parallel implementation of Algorithm $\dShortcut$ from \Cref{lem:dShortcut-parallel} to that graph, resulting in a $d$-shortcut $H^*_i$. The graph $G_{i+1}$ is obtained by adding to $G_i$ the shortcut edges mapped from $H^*_i$.
Formally, the algorithm repeats the following procedure for $\ell=O(d\log n)$ steps, where in step $i$ it applies the following on $G_i$:

\begin{itemize}
\item Create a $(d+2)$-layered graph $G^*_i=(L_1 \cup \ldots \cup L_{d+2}, E^*_i)$, where each level $L_j$ consists of copies of $V(G_i)$, that is, $L_j=\{v^j_1,\ldots, v^j_n\}$. The edges between consecutive layers $L_j$ and $L_{j+1}$ correspond to the edges in $G_i$: $E_{i,j}=\{(v^j_a,v^{j+1}_b) ~\mid~ (v_a,v_b)\in G_i\}$ and $E^*_i=\bigcup_{j=1}^{d+1}E_{i,j}$. 

\item Apply the parallel implementation of Alg. $\dShortcut$ to $G^*_i$ and let $H^*_i$ be the output shortcut. Let $H_i$ be the immediate mapping of the $H^*_i$-edges to edges in $V(G)\times V(G)$. Formally, $H_i=\{(v_a,v_b) ~\mid~ (v^j_a,v^{j'}_b)\in H^*_i \text{ for some } j,j'\in \{1,\ldots,d+2\}\}$.

\item Let $G_{i+1}=G_i \cup H_i$.

\end{itemize}

\textbf{Analysis.} Informally, we will show that each step reduces the current diameter by a factor of roughly $(1-1/(d+1))$ and hence within $O(d\log n)$ steps the diameter of the shortcut graph becomes at most $d$. Let $H_0=\emptyset$ and for every $i \in \{0,\ldots, \ell\}$, let $\hat{H}_i=\bigcup_{j=1}^i H_j$. Set $H=\hat H_\ell$. 


\begin{claim}\label{cl:diameter-bound-parallel}
$H$ is a $d$-shortcut.
\end{claim}
\begin{proof}
For every $j \in \{1,\ldots, \ell+1\}$, let $t_j$ be the $u$-$v$ distance in $G_j=G \cup \hat{H}_{j-1}$. For $t_j\geq d+1$, we show that $t_{j+1}\leq (1-1/(d+1))\cdot t_j+1$. 
By \Cref{lem:layered-shortcut-map}, every length-$(d+1)$ path in $G_j$ is shortened to length at most $d$ in $G_j\cup H_j$.
Partition a $t_j$-edge shortest $u$-$v$ path in $G_j$ into a maximal number of edge-disjoint segments of length $d+1$. Shortening every full segment to length at most $d$ gives

$$t_{j+1}\leq d\left\lfloor \frac{t_j}{d+1} \right\rfloor+(t_j \bmod (d+1))\leq\left(1-\frac{1}{d+1}\right)\cdot t_j+1~.$$

Solving this recurrence yields $t_j\leq d+1+(t_1-d-1)(1-1/(d+1))^{j-1}$. Choose $\ell=\Theta(d\log n)$ large enough that the final term is smaller than $1$. Since distances are integral, this gives $t_\ell\leq d+1$; one additional iteration gives $t_{\ell+1}\leq d$. Hence $\dist_{G \cup H}(u,v)=\dist_{G_{\ell+1}}(u,v)\leq d$.
\end{proof}

For every $i$, we have $D(G^*_i)\leq d+1$; by \Cref{lem:diameter-monotonicity}, every intermediate graph created while processing $G_i^*$ has diameter at most $d+1$ as well. Moreover, each reachable pair of $G_i$ induces at most $O(d^2)$ reachable pairs among its layered copies, so $|TC(G_i^*)|=O(d^2T)=\widetilde O(T)$. Thus \Cref{thm:dshortcutparallel} follows from the $O(d\log n)=\widetilde O(1)$ applications of \Cref{lem:dShortcut-parallel}.

\section{$d$-Shortcuts for $d=\Omega(\log^2 n)$}\label{sec:large-d-shortcut}

In this section, we prove \Cref{thm:d-largeshortcutparallel}. The main procedure shortens every shortest directed path of length $d$ to length at most $d/2$. Repeating this procedure $O(\log n)$ times on the current augmented graph yields a $d$-shortcut. To keep the depth bounded by $\widetilde O(d)$, the procedure restricts attention to incoming and outgoing neighborhoods of radius $\widetilde O(d)$. We may assume that $d<n$, since otherwise the empty shortcut suffices.

For an integer $r\geq0$, define the $r$-incoming neighborhood of a vertex $v$ by
\[
 \InPredBall{r}{v}{G}:=\{u\in V(G)\mid \dist_G(u,v)\leq r\},
\]
and let $\TCInDeg(v,r,G):=|\InPredBall{r}{v}{G}|$. Define the $r$-outgoing neighborhood $\OutSuccBall{r}{v}{G}$ and the $r$-outdegree $\TCOutDeg(v,r,G)$ symmetrically. As in \Cref{sec:prelim}, let $\TCDeg(v,r,G)$ denote the sum of the $r$-indegree and $r$-outdegree of $v$, and let $\TCDeg(r,G)$ be the maximum total $r$-degree over all vertices. Finally, let $T_r(G)$ denote the sum of the total $r$-degrees. When $G$ is clear from context, we omit it from this notation. Since every radius-$r$ neighborhood contains its center, $T_r(G)\geq2n$.

We extend the TC-hitting set definition of \Cref{def:TChitset} to the $d$-bounded neighborhoods in the following immediate manner.  

\begin{definition}[Balanced $d$-Hitting Sets]\label{def:TChitsetdd}
For a given $n$-vertex graph $G=(V,E)$, radius $d$, and integer $\Delta$, a subset $S \subseteq V$ is a $(d,\Delta)$-\emph{balanced TC-hitting set} if:
\begin{enumerate}
\item $S$ hits every $d$-outgoing and $d$-incoming neighborhood of size at least $\Delta$. That is, for every $v$ with $\TCOutDeg(v,d,G)\geq \Delta$, we have $S \cap \OutSuccBall{d}{v}{G} \neq \emptyset$, and for every $v$ with $\TCInDeg(v,d,G)\geq \Delta$, we have $S \cap \InPredBall{d}{v}{G}\neq \emptyset$.

\item $\sum_{s \in S} \left(\TCInDeg(s,d,G)+\TCOutDeg(s,d,G)\right)=\widetilde O(T_d(G) /\Delta)$.
\end{enumerate}
\end{definition}

We first record the truncated form of the size estimator that will also be used to construct these hitting sets.

\begin{lemma}\label{lem:estimate-up-to-d}
Given an $n$-vertex graph $G=(V,E)$ and an integer $d$, there is a randomized parallel algorithm that obtains constant-factor estimates of all $d$-incoming and $d$-outgoing degrees, and hence of $T_d(G)$, with $\widetilde O(|E|)$ work and $\widetilde O(d)$ depth.
\end{lemma}
\begin{proof}
Apply the estimator of Cohen~\cite{Cohen97} to both $G$ and $G^{\mathsf R}$, truncating every reachability exploration after $d$ rounds. The two applications estimate, respectively, all $d$-outgoing and $d$-incoming degrees. Truncation reduces the depth to $\widetilde O(d)$ and does not increase the work.
\end{proof}

The proof of \Cref{cl:computeTC-hit-set} now applies to the bounded neighborhoods; we include the details to make the role of randomness explicit.

\begin{claim}\label{cl:computeTC-hit-set-d}
Let $0\leq r\leq M$. There is a randomized near-linear-time algorithm $\TCHittingSetd$ that, given an $n$-vertex graph $G$ and parameters $r,M,\Delta$, outputs, with high probability, a set $S\subseteq V(G)$ that (i) hits every $r$-incoming and $r$-outgoing neighborhood of size at least $\Delta$, and (ii) satisfies $\sum_{s\in S}(\TCInDeg(s,M,G)+\TCOutDeg(s,M,G))=\widetilde O(T_M(G)/\Delta)$. The algorithm has a parallel implementation with near-linear work and $\widetilde O(M)$ depth. In particular, setting $r=M=d$ gives a $(d,\Delta)$-balanced TC-hitting set in the sense of \Cref{def:TChitsetdd}; the algorithm below uses the mixed-radius form.
\end{claim}
\begin{proof}
Generate $q=\Theta(\log n)$ candidate sets $S_1,\ldots,S_q$ by including every vertex independently with probability $p=\min\{1,c\log n/\Delta\}$, for a sufficiently large constant $c$. A union bound over all $r$-incoming and $r$-outgoing neighborhoods of size at least $\Delta$ shows that every candidate satisfies property~(1) with high probability. Moreover,
\[
 \mathbb E\!\left[\sum_{s\in S_j}
 \bigl(\TCInDeg(s,M,G)+\TCOutDeg(s,M,G)\bigr)\right]
 =pT_M(G).
\]
By Markov's inequality, each candidate satisfies property~(2) with constant probability. Thus, with high probability, at least one candidate satisfies both properties. Using \Cref{lem:estimate-up-to-d} with radius $M$, select the candidate of minimum estimated total $M$-degree. Its actual total $M$-degree is within a constant factor of the minimum and is therefore $\widetilde O(T_M(G)/\Delta)$. The sampling and selection require near-linear work, and the degree estimation has $\widetilde O(M)$ depth.
\end{proof}

\begin{lemma}\label{lem:naive-shortcut}
There is an algorithm $\NaiveShortcut$ that, given an $n$-vertex $m$-edge graph $G$, a topological ordering of its strongly connected components, and an integer $d\geq6$, computes a $d$-shortcut of size $\widetilde O\!\left(\frac{n^{\omega}}{d^{\omega-1}}\right)$ in time $\widetilde O\!\left(\frac{n^{\omega}}{d^{\omega-1}}+m\right)$. It has a parallel implementation with the same work bound and polylogarithmic depth.
\end{lemma}
\begin{proof}
We may assume that $d<n$, since otherwise the empty shortcut is already a $d$-shortcut.
Let $V_1,\ldots,V_q$ be the strongly connected components in the given topological order, and choose a representative $v_i\in V_i$ for every $i$. Add both directed edges between $v_i$ and every other vertex of $V_i$; let $H_0$ be the resulting set. Thus, any two vertices in one component are at distance at most two in $G\cup H_0$.

Let $G'$ be the DAG on representatives in which $(v_i,v_j)\in E(G')$ whenever $G$ has an edge from $V_i$ to $V_j$. Partition the topologically ordered vertices of $G'$ into $k=\min\{q,\lfloor d/6\rfloor\}$ consecutive blocks $Z_1,\ldots,Z_k$ of size $O(n/d)$. For every block, compute $TC(G'[Z_j])$, and let
\[
 H_1:=\bigcup_{j=1}^k TC(G'[Z_j]).
\]
This takes $\widetilde O\!\left(\frac{n^{\omega}}{d^{\omega-1}}\right)$ time in total.

We first verify that all added edges are valid shortcuts for the original graph. Every edge of $H_0$ joins two vertices in one strongly connected component and hence belongs to $TC(G)$. Moreover, if $(v_i,v_j)$ is an edge of $G'$, some original edge $(x,y)$ has $x\in V_i$ and $y\in V_j$; the path from $v_i$ to $x$ inside $V_i$, followed by $(x,y)$ and a path from $y$ to $v_j$ inside $V_j$, realizes $(v_i,v_j)$ in $G$. Therefore every path in $G'$, and in particular every edge of $H_1$, represents a reachable pair in $G$. Thus $H_0\cup H_1\subseteq TC(G)$.

Consider a path $P$ in $G'$. Label each vertex of $P$ by the index of its block. These labels are nondecreasing, so $P$ decomposes into $r\leq k$ consecutive subpaths, each contained in one block. Each such subpath is replaced by one edge of $H_1$. A transition edge $(v_i,v_j)$ between consecutive subpaths corresponds to some original edge $(x,y)$ with $x\in V_i$ and $y\in V_j$; it can therefore be realized in $G\cup H_0$ by the three-edge path $v_i\to x\to y\to v_j$. Including the two edges that connect the original endpoints to their component representatives, every reachable pair in $G$ is at distance at most
\[
 r+3(r-1)+2=4r-1\leq4k-1\leq d.
\]
Hence $H_0\cup H_1$ is a $d$-shortcut for $G$. Its size is $O(n)+\widetilde O\!\left(\frac{n^{\omega}}{d^{\omega-1}}\right)=\widetilde O\!\left(\frac{n^{\omega}}{d^{\omega-1}}\right)$. Constructing $G'$ takes $O(m)$ work. All block products can be computed in parallel using \Cref{thm:parrecmat}, which also gives the stated parallel bound.
\end{proof}

The following parallel LDD procedure takes a central role in our algorithm.
\begin{lemma}[Parallel LDD, \cite{BringmannFHL25,HaeuplerJS26}]\label{parallel-LDD}
Given an unweighted directed graph $G=(V,E)$ and a diameter parameter $d$, we can compute a set of edges $E^{\mathrm{rem}}$ and a list of vertex sets $V_1,\ldots, V_\ell$ such that:
\begin{itemize}
\item $V_1,\ldots, V_\ell$ are the strongly connected components of $G\setminus E^{\mathrm{rem}}$ in topological order; that is, an edge of $G\setminus E^{\mathrm{rem}}$ can point from $V_i$ to $V_j$ only if $i\leq j$,
\item each $V_i$ has a weak diameter at most $d$,
\item for every $e \in E$, we have $\Pr[e \in E^{\mathrm{rem}}]=O(\log^2 n/d)$.
\end{itemize}
The algorithm has $\tilde{O}(m)$ work and $\tilde{O}(d)$ depth.
\end{lemma}
 
\begin{corollary}\label{cor:LDD-cover}
There is an algorithm that, given an unweighted directed graph $G=(V,E)$ and a diameter parameter $d$, computes $k=\widetilde O(1)$ subgraphs $\mathcal{G}=\{G_1,\ldots,G_k\}$, together with a topological ordering of the SCCs of each $G_i$, such that, with high probability, every pair $u,v\in V$ with $\dist_G(u,v)\leq d$ satisfies $\dist_{G_i}(u,v)=\dist_G(u,v)$ for some $G_i\in\mathcal G$. The algorithm has $\widetilde O(m)$ work and $\widetilde O(d)$ depth.
\end{corollary}
\begin{proof}
Independently run \Cref{parallel-LDD} $c_1\log n$ times with decomposition diameter $D=c_2d\log^2 n$, for sufficiently large constants $c_1,c_2$, and let each $G_i$ be the graph obtained after deleting the edges removed in the corresponding run. Fix a pair $u,v$ with $\dist_G(u,v)\leq d$ and fix one shortest $u$-to-$v$ path $P$. This path has at most $d$ edges. By a union bound over its edges, a single run preserves all of $P$ with probability at least
\[
 1-|P|\cdot O(\log^2 n/D)\geq 1/2.
\]
Thus, the probability that none of the runs preserves $P$ is $n^{-\Omega(c_1)}$. A union bound over all ordered pairs proves the simultaneous guarantee with high probability. The runs are performed in parallel; their logarithmic multiplicity and the factor $\log^2 n$ in $D$ are absorbed by $\widetilde O(\cdot)$.
\end{proof}

\begin{lemma}\label{lem:parallel-in-ball}
There is a parallel algorithm $\ShortcutInBalls$ that, given an $n$-vertex graph $G$, a vertex $v\in V(G)$, and integers $d<M$, returns a shortcut set $H_v$ and the set $Q'_v=\InPredBall{M-d}{v}{G}$ such that, for every $x\in Q'_v$ and every $y\in\InPredBall{d}{x}{G}$,
\[
 \dist_{G\cup H_v}(y,x)\leq d/\log^2 n.
\]
The work is
\[
 \widetilde O\!\left(\frac{\TCInDeg(v,M,G)^{\omega}}{d^{\omega-1}}
 +\left|E\!\left(G[\InPredBall{M}{v}{G}]\right)\right|\right),
\]
and the depth is $\widetilde O(M)$.
Moreover, $|H_v|=\widetilde O\!\left(\frac{\TCInDeg(v,M,G)^{\omega}}{d^{\omega-1}}\right)$.
\end{lemma}
\begin{proof}
Let $Q_v=\InPredBall{M}{v}{G}$ and $Q'_v=\InPredBall{M-d}{v}{G}$. These sets can be computed with near-linear work and $\widetilde O(M)$ depth. Apply \Cref{cor:LDD-cover} to $G[Q_v]$ with diameter parameter $d$, obtaining $\widetilde O(1)$ topologically ordered subgraphs $\mathcal G_v=\{G_1,\ldots,G_k\}$. For every $G_i$, apply \Cref{lem:naive-shortcut} with $d'=\lfloor d/\log^2 n\rfloor$, and let $H_{i,v}$ be the resulting $d'$-shortcut. The constant implicit in $d=\Omega(\log^2 n)$ is chosen so that $d'\geq6$; smaller values can be handled by adjusting the constant in the theorem. Return
\[
 H_v=\bigcup_{i=1}^k H_{i,v}
 \qquad\text{and}\qquad Q'_v.
\]

\noindent \textbf{Correctness.} Fix $x\in Q'_v$ and $y\in\InPredBall{d}{x}{G}$. Since $\dist_G(y,x)\leq d$ and $\dist_G(x,v)\leq M-d$, the entire shortest $y$-to-$x$ path lies in $G[Q_v]$, and hence $\dist_{G[Q_v]}(y,x)=\dist_G(y,x)$. By \Cref{cor:LDD-cover}, some $G_i\in\mathcal G_v$ preserves this path. Therefore, \Cref{lem:naive-shortcut} gives $\dist_{G_i\cup H_{i,v}}(y,x)\leq d'$, as required.

\smallskip
\noindent \textbf{Work and Depth.} Let $m_v=|E(G[Q_v])|$ and $n_v=|Q_v|=\TCInDeg(v,M,G)$. By \Cref{cor:LDD-cover}, constructing $\mathcal G_v$ uses $\widetilde O(m_v)$ work and $\widetilde O(d)$ depth. Across the $\widetilde O(1)$ subgraphs, \Cref{lem:naive-shortcut} uses $\widetilde O\!\left(\frac{n_v^{\omega}}{d^{\omega-1}}+m_v\right)$ work. Computing $Q_v$ and $Q'_v$ has depth $\widetilde O(M)$, which dominates the depth bound.
\end{proof}


\begin{theorem}\label{thm:TC-large-d-deg}
Let $G$ be an $n$-vertex $m$-edge digraph, let $d<M$, and let $1\leq\Delta_{\mathrm{out}}\leq\Delta_{\max}$. Suppose that $\TCDeg(M,G)\leq\Delta_{\max}$, and let $T^*\geq T_M(G)$ be an upper bound. There is an algorithm $\ShortcutLargeDeg$ that returns a shortcut $H$ and a set $V'\subseteq V(G)$ such that every vertex $v$ with $\TCOutDeg(v,M-d,G)\geq\Delta_{\mathrm{out}}$ belongs to $V'$. Moreover, for every $v\in V'$ and every $u\in\InPredBall{d}{v}{G}$,
\[
 \dist_{G\cup H}(u,v)\leq d/\log^2 n.
\]
When $T^*/\Delta_{\mathrm{out}}<\Delta_{\max}$, the total work is 
$\widetilde O\left(\left(\frac{T^*}{\Delta_{\mathrm{out}} d}\right)^{\omega-1}\frac{T^*}{\Delta_{\mathrm{out}}}+m\right)$, and the output has size at most the first term. Otherwise, the work bound is $\widetilde O\left(\left(\frac{\Delta_{\max}}{d}\right)^{\omega-1}\frac{T^*}{\Delta_{\mathrm{out}}}+\frac{\Delta_{\max}}{\Delta_{\mathrm{out}}}\cdot m\right)$, and the output has size at most the first term. The depth in both cases is $\widetilde O(M)$.
\end{theorem}

\paragraph{Algorithm $\ShortcutLargeDeg$.}
First apply \Cref{cl:computeTC-hit-set-d} with inner radius $r=M-d$, outer radius $M$, and threshold $\Delta_{\mathrm{out}}$. Thus $S$ hits every $(M-d)$-incoming and $(M-d)$-outgoing neighborhood of size at least $\Delta_{\mathrm{out}}$, while its total $M$-degree is $\widetilde O(T^*/\Delta_{\mathrm{out}})$. We then distinguish two cases.

\noindent \textbf{Case $T^*/\Delta_{\mathrm{out}}<\Delta_{\max}$.}
Add a dummy sink $w$ and an edge $(s,w)$ for every $s\in S$, obtaining
\[
 G'=(V\cup\{w\},E\cup\{(s,w):s\in S\}).
\]
Apply $\ShortcutInBalls$ to $\langle G',w,M+1,d\rangle$, and let $(H_w,V_w)$ be its output. Return
\[
 H=H_w\cap(V\times V)
 \qquad\text{and}\qquad
 V'=V_w\cap V.
\]

\noindent \textbf{Otherwise.}
For every $s\in S$, apply $\ShortcutInBalls$ to $\langle G,s,M,d\rangle$ in parallel, and let $(H_s,V_s)$ be its output. Return $H=\bigcup_{s\in S}H_s$ and $V'=\bigcup_{s\in S}V_s$.

\noindent \textbf{Correctness.}
In the first case, every edge of $H_w$ belongs to $TC(G')$. If both endpoints are in $V$, a path in $G'$ realizing this edge cannot use $w$, because $w$ is a sink. Hence $H=H_w\cap(V\times V)\subseteq TC(G)$. In the second case, every $H_s$ is returned as a shortcut for an induced subgraph of $G$, and therefore $H=\bigcup_{s\in S}H_s\subseteq TC(G)$.

Fix $v$ with $\TCOutDeg(v,M-d,G)\geq\Delta_{\mathrm{out}}$. By \Cref{cl:computeTC-hit-set-d}, some $s\in S$ lies in $\OutSuccBall{M-d}{v}{G}$, and hence $v\in\InPredBall{M-d}{s}{G}$.

In the first case, $\dist_{G'}(v,w)\leq M-d+1=(M+1)-d$, so $v\in V_w$. For every $u\in\InPredBall{d}{v}{G}$, \Cref{lem:parallel-in-ball} gives $\dist_{G'\cup H_w}(u,v)\leq d/\log^2 n$. Since $w$ is a sink, this $u$-to-$v$ path does not use $w$, and therefore it remains in $G\cup H$. In the second case, $v\in V_s$, and the same conclusion follows directly from the application centered at $s$.

\noindent \textbf{Work and Depth.}
Computing $S$ uses near-linear work and $\widetilde O(M)$ depth. In the first case,
\[
 \TCInDeg(w,M+1,G')
 \leq |S|+\sum_{s\in S}\TCInDeg(s,M,G)
 =\widetilde O(T^*/\Delta_{\mathrm{out}}).
\]
Moreover, $G'[\InPredBall{M+1}{w}{G'}]$ has at most $m+|S|$ edges. Thus, \Cref{lem:parallel-in-ball} gives
\[
 \widetilde O\!\left(
 \left(\frac{T^*}{\Delta_{\mathrm{out}}d}\right)^{\omega-1}
 \frac{T^*}{\Delta_{\mathrm{out}}}+m\right)
\]
work and $\widetilde O(M)$ depth.
The corresponding output-size bound follows from the last assertion of \Cref{lem:parallel-in-ball}.

Now suppose $T^*/\Delta_{\mathrm{out}}\geq\Delta_{\max}$. We first bound the total size of the induced subgraphs processed by the parallel calls:
\[
 \sum_{s\in S}\left|E\!\left(G[\InPredBall{M}{s}{G}]\right)\right|
 \leq \widetilde O\!\left(\frac{\Delta_{\max}}{\Delta_{\mathrm{out}}}\cdot m\right).
\]
Indeed, if an edge $(x,y)$ belongs to $G[\InPredBall{M}{s}{G}]$, then $s\in\OutSuccBall{M}{x}{G}$. The latter set has size at most $\Delta_{\max}$. The construction in \Cref{cl:computeTC-hit-set-d} samples each candidate set at rate $\widetilde O(1/\Delta_{\mathrm{out}})$; a Chernoff bound and a union bound over all vertices and all candidates show that the selected set $S$ satisfies
\[
 |S\cap\OutSuccBall{M}{x}{G}|=\widetilde O(\Delta_{\max}/\Delta_{\mathrm{out}})
\]
simultaneously for every $x$.

Since every $M$-indegree is at most $\Delta_{\max}$, \Cref{lem:parallel-in-ball} and \Cref{cl:computeTC-hit-set-d} give total work
\begin{align*}
 &\frac{1}{d^{\omega-1}}\sum_{s\in S}\TCInDeg(s,M,G)^{\omega}
 +\sum_{s\in S}\left|E\!\left(G[\InPredBall{M}{s}{G}]\right)\right|\\
 &\qquad\leq
 \widetilde O\!\left(
 \left(\frac{\Delta_{\max}}{d}\right)^{\omega-1}\frac{T^*}{\Delta_{\mathrm{out}}}
 +\frac{\Delta_{\max}}{\Delta_{\mathrm{out}}}\cdot m\right).
\end{align*}
The same calculation without the induced-subgraph edge term bounds $|H|$ by the first term. 
All calls run in parallel, so the depth is $\widetilde O(M)$.

\paragraph{A Diameter-Halving Phase.}
We now give a procedure that shortens every shortest path of length $d$ to length at most $d/2$. Let $L=\lceil\log(16n/d)\rceil$ and set $M_1=(L+2)d$. Use \Cref{lem:estimate-up-to-d}, scaling its estimate by a constant if necessary, to compute an upper estimate $T^*$ satisfying $T_{M_1}(G)\leq T^*\leq2T_{M_1}(G)$. If $G$ contains a shortest path of length $d$, then the ordered pairs along this path imply $T^*\geq(d+1)(d+2)>d^2$. Consequently, if $T^*<d^2$, the algorithm returns the empty shortcut. In the remainder, assume $T^*\geq d^2$ and set
\[
 k=\left\lceil\log\left(\frac{8\sqrt{T^*}}{d}\right)\right\rceil,
 \qquad \Delta_0=n,
 \qquad \Delta_i=\left\lfloor\frac{\sqrt{T^*}}{2^i}\right\rfloor\quad (i\geq1).
\]
Since $T^*\leq4n^2$, we have $k\leq L$, and hence $M_{k+1}\geq2d$ for the sequence $M_{i+1}=M_i-d$. Moreover, $1\leq\Delta_i\leq\Delta_{i-1}$ for every $i\in\{1,\ldots,k\}$. Throughout the algorithm, $G_i$ is an induced subgraph of $G$ and $M_i\leq M_1$, so $T_{M_i}(G_i)\leq T_{M_1}(G)\leq T^*$; the same holds after reversing all edges. The total-degree precondition is established inductively in \Cref{cl:degree-reduc}.

\begin{wrapper}\vspace{-2pt}
\begin{center}\textbf{Algorithm $\LargedShortcut(G,d)$: Diameter-Halving Phase}
\end{center}
\textbf{Input}: An $n$-vertex graph $G=(V,E)$ and an integer $d=\Omega(\log^2 n)$.\\
\textbf{Output}: A shortcut $H$ such that every pair $u,v$ with $\dist_G(u,v)=d$ satisfies $\dist_{G\cup H}(u,v)\leq d/2$.
\begin{enumerate}
\item Compute $M_1$ and $T^*$ as above. If $T^*<d^2$, return the empty shortcut. Otherwise, set $G_1=G$ and initialize $k,\Delta_0,\ldots,\Delta_k$ as above.
\item For $i\in\{1,\ldots,k\}$ do:
\begin{enumerate}
\item $(H_{i,\mathrm{out}},V_{i,\mathrm{out}})=\ShortcutLargeDeg(G_i,M_i,d,2\Delta_{i-1},\Delta_i,T^*)$.
\item $(H'_{i,\mathrm{in}},V_{i,\mathrm{in}})=\ShortcutLargeDeg(G_i^{\mathsf R},M_i,d,2\Delta_{i-1},\Delta_i,T^*)$, and set $H_{i,\mathrm{in}}=(H'_{i,\mathrm{in}})^{\mathsf R}$.
\item $G_{i+1}=G_i[V(G_i)\setminus(V_{i,\mathrm{out}}\cup V_{i,\mathrm{in}})]$.
\item $M_{i+1}=M_i-d$.
\end{enumerate}
\item Return $H=\bigcup_{i=1}^k(H_{i,\mathrm{out}}\cup H_{i,\mathrm{in}})$.
\end{enumerate}
\end{wrapper}

\begin{claim}\label{cl:degree-reduc}
For every $i\in\{1,\ldots,k+1\}$, $\TCDeg(M_i,G_i)\leq2\Delta_{i-1}$.
\end{claim}
\begin{proof}
The claim holds for $i=1$ because every incoming and outgoing neighborhood has size at most $n=\Delta_0$. Suppose it holds for iteration $i$. Since reversal swaps incoming and outgoing degrees, both calls to $\ShortcutLargeDeg$ satisfy the upper-degree precondition with parameter $2\Delta_{i-1}$. By \Cref{thm:TC-large-d-deg}, $V_{i,\mathrm{out}}$ contains every vertex whose $(M_i-d)$-outdegree in $G_i$ is at least $\Delta_i$. Applying the same theorem to $G_i^{\mathsf R}$ shows that $V_{i,\mathrm{in}}$ contains every vertex whose $(M_i-d)$-indegree in $G_i$ is at least $\Delta_i$. Since $M_{i+1}=M_i-d$ and $G_{i+1}$ is an induced subgraph of $G_i$, every vertex remaining in $G_{i+1}$ has both $M_{i+1}$-indegree and $M_{i+1}$-outdegree less than $\Delta_i$, and hence total $M_{i+1}$-degree less than $2\Delta_i$.
\end{proof}

\begin{claim}\label{cl:path-partition}
Let $P$ be a shortest $u$-to-$v$ path of length $d$, and set $C=d/\log^2n+1$. After every $i\in\{0,\ldots,k\}$ iterations, either $\dist_{G\cup H}(u,v)\leq2iC$, or there is a nonempty $u_{i+1}$-to-$v_{i+1}$ subpath $P_{i+1}\subseteq P\cap G_{i+1}$ such that
\[
 \dist_{G\cup H}(u,u_{i+1})+\dist_{G\cup H}(v_{i+1},v)\leq2iC.
\]
\end{claim}
\begin{proof}
The claim is immediate for $i=0$, with $P_1=P$. Suppose a residual path $P_i\subseteq G_i$ exists before iteration $i$.

If $P_i$ contains a vertex of $V_{i,\mathrm{out}}$, let $a$ be the last such vertex. The subpath from $u_i$ to $a$ lies in $G_i$ and has length at most $d$, because it is contained in the original length-$d$ path $P$. Hence $u_i\in\InPredBall{d}{a}{G_i}$, and \Cref{thm:TC-large-d-deg} shortens this prefix to at most $d/\log^2n$ edges. If $a=v_i$, combining this shortcut with the prefix and suffix accumulated in earlier iterations gives the first outcome of the claim. Otherwise, let $\bar u$ be the successor of $a$ on $P_i$; reaching $\bar u$ costs at most $C$ additional edges. If $P_i$ contains no vertex of $V_{i,\mathrm{out}}$, set $\bar u=u_i$ and incur no additional cost.

Now consider the suffix of $P_i$ from $\bar u$ to $v_i$. If it contains a vertex of $V_{i,\mathrm{in}}$, let $b$ be the first such vertex. This suffix also has length at most $d$. Applying \Cref{thm:TC-large-d-deg} to $G_i^{\mathsf R}$ and reversing the resulting shortcut therefore shows that $b$ reaches $v_i$ in at most $d/\log^2n$ edges. If $b=\bar u$, combining the new prefix and suffix shortcuts with those accumulated earlier gives the first outcome of the claim. Otherwise, let $\bar v$ be the predecessor of $b$ on $P_i$; the suffix from $\bar v$ to $v_i$ then costs at most $C$ edges. If no such $b$ exists, set $\bar v=v_i$ and incur no additional cost.

In the remaining case, define $P_{i+1}=P_i[\bar u,\bar v]$. By the choices of $a$ and $b$, this subpath contains no vertex of $V_{i,\mathrm{out}}\cup V_{i,\mathrm{in}}$, so $P_{i+1}\subseteq G_{i+1}$. The accumulated prefix and suffix cost increases by at most $2C$, proving the induction step.
\end{proof}

We now prove the phase guarantee. If the first outcome of \Cref{cl:path-partition} occurs, then $\dist_{G\cup H}(u,v)\leq2kC$. Otherwise, there is a residual path $P_{k+1}\subseteq G_{k+1}$. By \Cref{cl:degree-reduc}, $\TCDeg(M_{k+1},G_{k+1})\leq2\Delta_k\leq d/4$. Since $M_{k+1}\geq d$ and $P_{k+1}$ is a simple path of length at most $d$, the $M_{k+1}$-outgoing neighborhood of its first vertex contains all $|P_{k+1}|+1$ vertices of the path. Hence $|P_{k+1}|<2\Delta_k\leq d/4$. Moreover, $k=O(\log n)$ and $d=\Omega(\log^2n)$, so $2kC=O(d/\log n)$. Thus, by choosing the constant implicit in $d=\Omega(\log^2n)$ sufficiently large,
\[
 \dist_{G\cup H}(u,v)\leq2kC+d/4\leq d/2.
\]

\paragraph{Work and Depth of a Diameter-Halving Phase.}
The estimate $T^*$ is computed with near-linear work and $\widetilde O(d)$ depth because $M_1=\widetilde O(d)$. In the first iteration, either case of \Cref{thm:TC-large-d-deg}, with upper-degree parameter $2\Delta_0=2n$, gives
\[
 \widetilde O\!\left(\frac{(T^*)^{\omega/2}}{d^{\omega-1}}+m\right)
\]
work: if its second case applies, then $n=O(\sqrt{T^*})$. For every $i\geq2$, we have $\Delta_i=\Theta(\Delta_{i-1})$ and $T^*/\Delta_i\geq2\Delta_{i-1}$. Hence the second case of \Cref{thm:TC-large-d-deg}, with upper-degree parameter $2\Delta_{i-1}$, gives
\[
 \widetilde O\!\left(
 \left(\frac{2\Delta_{i-1}}d\right)^{\omega-1}\frac{T^*}{\Delta_i}
 +\frac{2\Delta_{i-1}}{\Delta_i}\cdot m\right)
 =\widetilde O\!\left(\frac{(T^*)^{\omega/2}}{d^{\omega-1}}+m\right).
\]
The two calls in each iteration run in parallel. Since $k=O(\log n)$ and $M_i=\widetilde O(d)$, one phase uses $\widetilde O\!\left(\frac{(T^*)^{\omega/2}}{d^{\omega-1}}+m\right)$ work and $\widetilde O(d)$ depth. The output-size guarantees in \Cref{thm:TC-large-d-deg} also show that one phase adds at most $\widetilde O\!\left(\frac{(T^*)^{\omega/2}}{d^{\omega-1}}\right)$ edges.

\paragraph{Completing the Proof of \Cref{thm:d-largeshortcutparallel}.}
Run the phase repeatedly on the current augmented graph. All added edges belong to the original transitive closure, so the value of $T=|TC(G)|$ does not change. In every phase, $T_{M_1}(G)=O(T)$, and hence the estimate used by that phase satisfies $T^*=O(T)$. Consider a shortest path of length $q d+r$, where $0\leq r<d$. Partition it into $q$ subpaths of length $d$ and one residual subpath of length $r$. Each length-$d$ subpath is itself shortest, so after one phase the entire path has length at most $qd/2+r$. Thus, distances of at least $2d$ decrease by a constant factor, while a distance in $(d,2d)$ decreases by at least $d/2$. After $O(\log n)$ phases, every reachable pair is therefore at distance at most $d$.

The edge term in a phase is the number of edges in the current augmented graph. By the output-size bound above, after $j$ phases this number is at most $m+\widetilde O\!\left(\frac{jT^{\omega/2}}{d^{\omega-1}}\right)$. Therefore, summing these edge terms over $O(\log n)$ phases incurs only another polylogarithmic factor. This factor, as well as the number of phases, is absorbed by $\widetilde O(\cdot)$, yielding total work
\[
 \widetilde O\!\left(\frac{T^{\omega/2}}{d^{\omega-1}}+m\right)
\]
and depth $\widetilde O(d)$. By choosing the failure probabilities in the preceding randomized subroutines to be $n^{-c}$ for a sufficiently large constant $c$, a union bound makes all their guarantees hold simultaneously over all iterations and phases.

\section{$3$-Shortcuts}\label{sec:3-shortcut}

In this section we prove the $d=3$ case of \Cref{thm:dshortcut} by providing a $3$-shortcut construction with running time $\tilde{O}(T^{5/4})$ when $\omega=2$.

\begin{lemma}\label{lem:shortcut-two-len-paths}
Given a graph $G$ and vertex sets $A,B$ with $|B|\geq|A|$, algorithm $\ShortcutTwoPaths$ runs in $O(|A|^{\omega-1}|B|)$ time and returns a set $H\subseteq TC(G)$ such that $(a,b)\in H$ whenever $a\in A$, $b\in B$, and $G$ contains a path $a\to c\to b$ with $c\in A$.
\end{lemma}
\begin{proof}
Let $M_A\in\{0,1\}^{A\times A}$ and $M_B\in\{0,1\}^{A\times B}$ be defined by $M_A(a,c)=1$ iff $(a,c)\in E(G)$ and $M_B(c,b)=1$ iff $(c,b)\in E(G)$. Compute their Boolean product $M_C=M_AM_B$ and return $H=\{(a,b)\in A\times B\mid M_C(a,b)=1\}$. Partitioning the columns of $M_B$ into blocks of size $|A|$ gives running time $O(|A|^{\omega-1}|B|)$. Moreover, $M_C(a,b)=1$ exactly when some $c\in A$ satisfies $(a,c),(c,b)\in E(G)$. Hence every required edge is added, and every edge of $H$ belongs to $TC(G)$.
\end{proof}

%
%
%
%
%

\begin{mdframed}[hidealllines=true, backgroundcolor=gray!20, leftmargin=0cm,innerleftmargin=0.5cm,innerrightmargin=0.5cm,innertopmargin=0.5cm,innerbottommargin=0.5cm,roundcorner=10pt]
\vspace{-2pt}
\begin{center}\textbf{Algorithm $\ThreeShortcut(G)$}
\end{center}

\paragraph{1. Handling vertices with TC-degrees $O(T^{1/4})$.}

\begin{enumerate}
\item Compute the constant-factor upper estimate $T$ of $|TC(G)|$ and the simultaneous upper estimates $\{\hat{d}_{in}(v),\hat{d}_{out}(v)\}_{v}$ described after \Cref{thm:size-estimation}.
\item Let $V_{low}=\{v \in V ~\mid~ \hat{d}_{in}(v)+\hat{d}_{out}(v)\leq T^{1/4}\}$.
\item Compute $\{\OutSucc{G}{v},\InPred{G}{v}\}_{v \in V_{low}}$.
\item Let $H^{in}_{low}=\{(u,v) ~\mid~ v \in V_{low}, u\in \InPred{G}{v}\}$ and $H^{out}_{low}=\{(v,u) ~\mid~ v \in V_{low},u\in \OutSucc{G}{v}\}$. 
\item $H_{low}=H^{in}_{low} \cup H^{out}_{low}$.
\end{enumerate}

\paragraph{2. Handling the remaining $O(T^{3/4})$ vertices.}

\begin{enumerate}
\item $V_{high}=V \setminus V_{low}$ and $G_{2}=G[V \setminus V_{low}]$.
\item Set $\DeltaInInd_2=T$, $\DeltaOutInd_2=T^{1/2}$, $\DeltaInInd_1=T^{1/2}$, $\DeltaOutInd_1=T^{1/4}$, and $\DeltaInInd_0=T^{1/4}$.
\item For $i \in \{2,1\}$ do:
\begin{enumerate}
\item $TC_i=\TCLargeDeg(G_i, \DeltaInInd_{i}, \DeltaOutInd_{i})$ and $V_i=V(TC_i)$.
\item $G_{i-1}=(G_i\setminus V(TC_i))^{\mathsf{R}}$.
\end{enumerate}
\item $TC_0=\TCNaiveComb(G_{0},\DeltaInInd_{0})$.
\item For every $v \in V(G_1)$: $E_{low}(v)=\{(v,w) ~\mid~ \exists u \in N_{out}(v,G_1)\cap V(G_0) \mbox{~and~} (w,u) \in TC_0\}$. 
\item $E_{0 \to 1}=\bigcup_{v \in V(G_1)} E^{R}_{low}(v)$.
\item $E_{2 \to 0}=\ShortcutTwoPaths(G \cup TC_2,V_2,V_{high})$.
\item Return $H=H_{low}\cup TC_2 \cup TC^R_1 \cup TC_0 \cup E_{0 \to 1} \cup E_{2 \to 0}$.
\end{enumerate}
\end{mdframed}

\noindent\textbf{Correctness.} First observe that $H \subseteq TC(G)$. This is immediate for $H_{low}$, $TC_2$, $TC_1^{\mathsf R}$, $TC_0$, and $E_{2\to0}$. Moreover, every $(v,w)\in E_{low}(v)$ satisfies $(v,u)\in E(G_1)$ and $(w,u)\in TC_0$ for some $u\in V(G_0)$. Since $G_0\subseteq G_1^{\mathsf R}$, we have $(u,w)\in TC(G_1)$, and hence $(v,w)\in TC(G_1)$. Reversing these edges shows that $E_{0\to1}\subseteq TC(G)$. We therefore focus on the diameter argument. Let $V_i=V(TC_i)$ for $i \in \{0,1,2\}$. Consider a $u$-$v$ path $P$ in $G$. Assume first that $V(P)\cap V_{low}\neq \emptyset$ and let $w \in V(P)\cap V_{low}$. It then holds that $(u,w) \in H^{in}_{low}$ and $(w,v)\in H^{out}_{low}$ and thus $\dist_{G \cup H}(u,v)\leq 2$. From now on, assume that $V(P)\subseteq V \setminus V_{low}$ and thus $P \subseteq G_2$.

Apply \Cref{lem:close-P} first to $P$ and $TC_2$. The residual $P\setminus V_2$ is a suffix of $P$, whose reversal is a path in $G_1$. Applying \Cref{lem:close-P} again, now to this reversed path and $TC_1$, shows that deleting $V_1$ leaves a suffix in the reversed orientation. Returning to the original orientation, the vertices of $P$ therefore occur in the order $V_2,V_0,V_1$, with any of the three blocks possibly empty; here $V_0=V(G_0)$ because $TC_0=TC(G_0)$ and reachability is reflexive.

By \Cref{thm:TC-large-deg}, $TC_i=TC(G_i)[V_i]$ for $i\in\{1,2\}$. Thus, after accounting for the reversal defining $G_1$, the endpoints of every nonempty block in $V_2,V_0,V_1$ are joined by an edge of $TC_2,TC_0,TC_1^{\mathsf R}$, respectively.

Assume first that $V_i \cap P=\emptyset$ for some $i \in \{0,1,2\}$. Then $P$ decomposes into at most two nonempty blocks $P_1,\ldots,P_r$, where $r\leq2$ and each block lies in one $V_{i_j}$. If $i_j$ is even, its endpoints are joined by an edge of $TC_{i_j}$; otherwise they are joined by an edge of $TC^R_{i_j}$. Including the at most one connecting edge gives $\dist_{G \cup H}(u,v)\leq2r-1\leq3$. 

We next consider the case where 
$P=P_1 \circ (v_1,u_2) \circ P_2 \circ (v_2,u_3) \circ P_3$ where each $P_i$ is a $u_i$-$v_i$ path, and $P_1 \subseteq V_2$, $P_3 \subseteq V_1$ and $P_2 \subseteq V_0$. We have that $(u_1,v_1)\in TC_2$, $(u_2,v_2)\in TC_0$ and $(u_3,v_3)\in TC^R_1$. Moreover, $(u_3,v_2)\in E(G_1)$ and $(u_2,v_2)\in TC_0$. Hence $(u_3,u_2)\in E_{low}(u_3)$ and thus $(u_2,u_3)\in E_{0 \to 1}$.  

It remains to show that also $(u_1,u_2)\in E_{2 \to 0}$ which will establish our claim as $(u_1,u_2)\circ (u_2,u_3)\circ (u_3,v_3)\subseteq G \cup H$ and thus $\dist_{G \cup H}(u_1,v_3)\leq3$, as desired. We apply \Cref{lem:shortcut-two-len-paths} with $A=V_2$ and $B=V \setminus V_{low}$; the size condition holds because $V_2\subseteq V\setminus V_{low}$. We have $u_1,v_1 \in A$ and $u_2 \in B$, and since $(u_1,v_1)\in TC_2$, the graph $G \cup TC_2$ contains the $2$-length $u_1$-$u_2$ path $[u_1,v_1,u_2]$.
Therefore $(u_1,u_2)\in E_{2 \to 0}$, and the diameter argument follows. See Fig. \ref{fig:diam3} for an illustration.

\noindent\textbf{Runtime.} By the proof of \Cref{lem:naive-comb-small-deg}, Step 1 of computing the shortcut set $H_{low}$ takes $O(T^{5/4})$ time. Since every vertex $v \in V \setminus V_{low}$ has $\hat{d}_{in}(v)+\hat{d}_{out}(v)=\Omega(T^{1/4})$, it holds that $|V \setminus V_{low}|=\widetilde O(T^{3/4})$. By \Cref{thm:TC-large-deg}, $TC_2$ is computed in $\widetilde O(T^{\omega/2})$ time, whereas $TC_1$ is computed in $\widetilde O(T^{(\omega-1)/2+3/4})$ time. Both bounds are $\widetilde O(T^{5\omega/8})$ for $\omega\geq2$. By \Cref{lem:naive-comb-small-deg}, the computation of $TC_0$ takes $\widetilde O(T^{5/4})$ time. Computing $E_{low}(v)$ takes $|N_{out}(v,G_1)|\cdot T^{1/4}$ time, and therefore computing $E_{0 \to 1}$ takes $O(T^{5/4})$ time.

It remains to bound the computation time of $E_{2 \to 0}$. In the construction of $TC_2$, the balanced hitting-set guarantee gives $|V_2|\leq\sum_{s\in S_2}\TCInDeg(s,G_2)=\widetilde O(T/\DeltaOutInd_2)=\widetilde O(\sqrt T)$. Since $|V_{high}|=\widetilde O(T^{3/4})$, \Cref{lem:shortcut-two-len-paths} computes $E_{2 \to 0}$ in $\widetilde O((\sqrt T)^{\omega-1}T^{3/4})=\widetilde O(T^{(\omega-1)/2+3/4})=\widetilde O(T^{5\omega/8})$ time, where the last inequality uses $\omega\geq2$. This concludes the proof of the $d=3$ case of \Cref{thm:dshortcut}.

\paragraph{Acknowledgements.}
We are grateful to Oded Goldreich for many valuable comments on this draft. We also thank Nick Fischer for helpful explanations concerning the conjectured running time for sparse matrix multiplication in~\cite{AbboudBFK24}.

\bibliographystyle{alpha}
\bibliography{thesis-codex}

\appendix
%
%
%
\section{Proof of \Cref{thm:joint-hardness}}\label{app:joint-hardness}

\begin{proof}
Fix $n\leq T\leq n^2$. We combine three components: a shortcut-hard graph, a conditional transitive-closure-hard graph, and a path that fixes the total closure size.

We first formalize the choice of $\epsilon$. Fix constants $\epsilon_j\downarrow0$. For every $j$, \Cref{thm:hesse} gives a constant $c_j>0$ and, for all sufficiently large $N$, an $N$-vertex graph $A_{N,j}$ such that every shortcut reducing its diameter to at most $j\log N$ has at least $c_jN^{2-\epsilon_j}$ edges; this follows because $j\log N=o(N^{\delta_j})$ for the corresponding constant $\delta_j>0$. Choose increasing thresholds $N_j$ so that these graphs exist for every $N\geq N_j$ and $c_jN^{\epsilon_j}\geq1$. Let $j(N)=\max\{j:N\geq N_j\}$ and set $A_N=A_{N,j(N)}$. Since $j(N)\to\infty$ and $\epsilon_{j(N)}\to0$, every $O(\log N)$-shortcut for $A_N$ has at least
\[
 c_{j(N)}N^{2-\epsilon_{j(N)}}\geq N^{2-2\epsilon_{j(N)}}=N^{2-o(1)}
\]
edges. Let $N_A=\lfloor c_A\sqrt T\rfloor$ and let $A=A_{N_A}$, where $c_A>0$ is a sufficiently small constant. Since $T\geq n$, we have $\log N_A=\Theta(\log n)$. Thus, every $O(\log n)$-shortcut for $A$ has
\[
 N_A^{2-o(1)}=T^{1-o(1)}
\]
edges, while $|TC(A)|=O(T)$.

For the computationally hard component, set
\[
 \tau=c_B\min\{T,n^{3/2}\},
\]
for a sufficiently small constant $c_B>0$. The $r=1$ case of the ABFK lower bound~\cite[Theorem~1.10]{AbboudBFK24}, together with the standard three-layer reduction from sparse Boolean matrix multiplication to transitive closure, gives a three-layer DAG $B$ with
\[
 |V(B)|=O(\tau^{2/3}),
 \qquad |TC(B)|=\Theta(\tau),
\]
such that, assuming the PS-AE-Triangle hypothesis, computing $TC(B)$ requires $\tau^{4/3-o(1)}$ time when $\omega=2$. The three vertex layers have sizes $\tau^{1/3},\tau^{2/3},\tau^{2/3}$, and standard padding ensures input and output size $\Theta(\tau)$. Moreover,
\[
 \tau^{4/3}=\Theta\!\left(\min\{T^{4/3},n^2\}\right)
 =\Theta(\mathsf{B}(n,T)).
\]

Finally, let $F$ be a directed path on $N_F=\lfloor c_F\sqrt T\rfloor$ vertices, for a sufficiently small constant $c_F>0$. Then $|TC(F)|=\Theta(T)$. Take the disjoint union of $A$, $B$, and $F$, add a new vertex $x$, an edge from $x$ to an arbitrary vertex of $A$, and an edge from $x$ to an arbitrary vertex of $B$. Call the resulting graph $G^*$. The two new edges add only $O(|V(A)|+|V(B)|)$ closure pairs. Consequently, $|TC(G^*)|=\Theta(T)$.

Also,
\[
 |V(A)|+|V(B)|+|V(F)|+1
 =O\!\left(\sqrt T+\min\{T^{2/3},n\}\right)\leq n
\]
when $c_A,c_B,c_F$ are chosen sufficiently small. We add isolated vertices to make the number of vertices exactly $n$.

No path with both endpoints in $A$ can leave $A$ and return. Hence the restriction to $A$ of any $O(\log n)$-shortcut for $G^*$ is an $O(\log n)$-shortcut for $A$, and therefore contains $T^{1-o(1)}$ edges.

It remains to establish the computational lower bound. The subgraph of $TC(G^*)$ induced by $V(B)$ is exactly $TC(B)$. Thus, an algorithm computing $TC(G^*)$ in $O(\mathsf{B}(n,T)^{1-\varepsilon})$ time, for any fixed $\varepsilon>0$, would also compute $TC(B)$ within this time by retaining the output pairs in $V(B)\times V(B)$. If this running time is $o(T)$, it already contradicts the time needed to write the $\Theta(T)$ output pairs; otherwise the filtering costs only the claimed running time. Since $\mathsf{B}(n,T)=\Theta(\tau^{4/3})$, either case contradicts the ABFK lower bound. Therefore computing $TC(G^*)$ requires $\mathsf{B}(n,T)^{1-o(1)}$ time under the PS-AE-Triangle hypothesis.
\end{proof}

\section{Figures}

\begin{figure}[h!]
\begin{center}
\includegraphics[width=\linewidth,trim=0 80 0 80,clip]{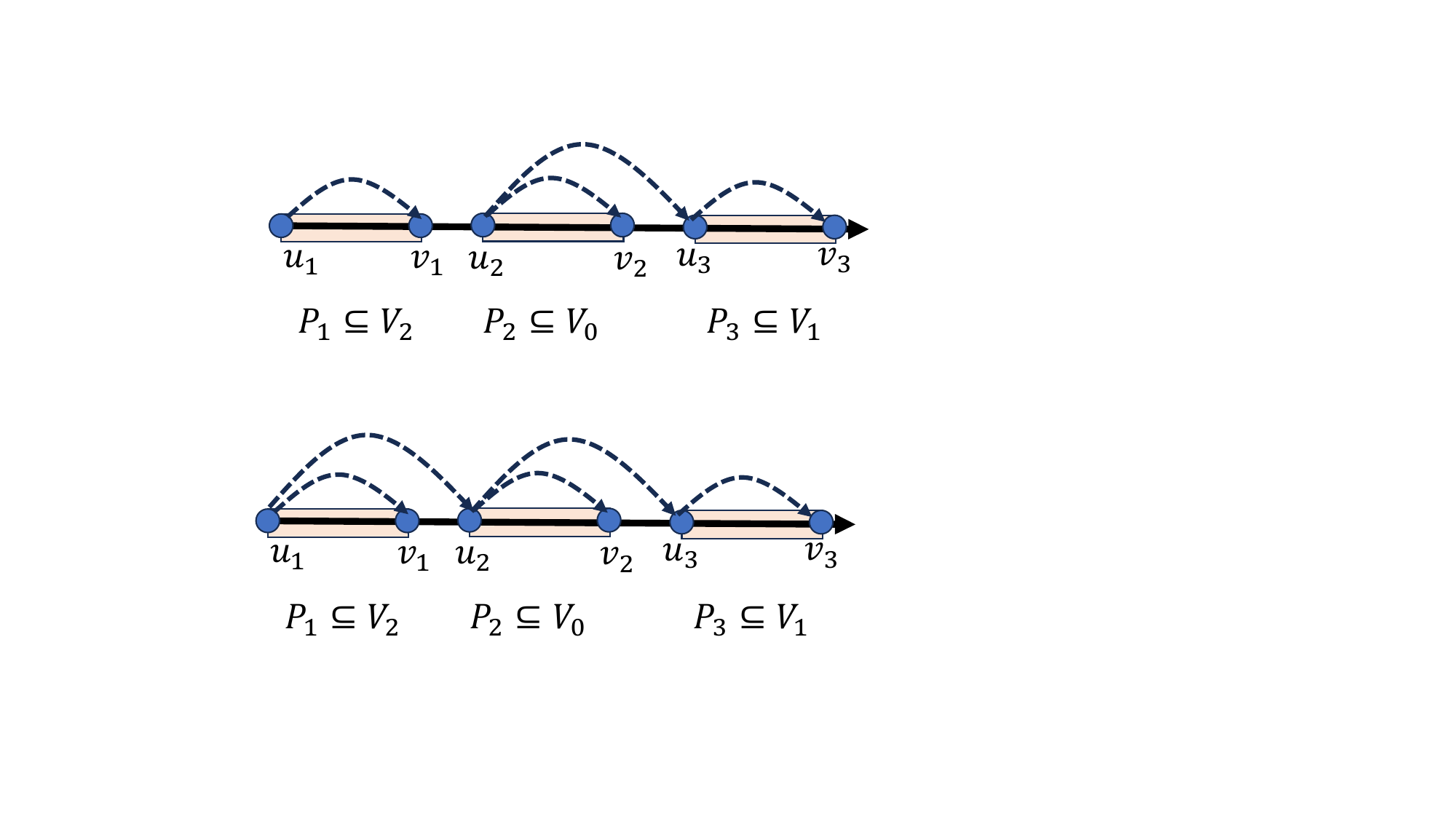}
\caption{\sf Illustration for the diameter $4$ (top) and $3$ (bottom) arguments underlying \Cref{thm:dshortcut}. Solid black edges belong to the original path, shaded boxes are the path blocks, and dashed blue edges are shortcuts; dashed edges that cross a block boundary are the stitching edges.
\label{fig:diam3} 
}
\end{center}
\end{figure}

%
%
%
%
%
%

\end{document}